\documentclass{llncs}

\usepackage[utf8]{inputenc}
\usepackage[english]{babel}
\usepackage[T1]{fontenc}

\usepackage{multicol}
\usepackage{listings}
\usepackage{graphicx}
\usepackage{hyperref}
\usepackage{amssymb}
\usepackage{calrsfs}
\usepackage{amsfonts}
\usepackage{amsmath,amssymb}
\usepackage{mathtools}
\usepackage{varioref}
\usepackage{xcolor}
\usepackage{comment}
\usepackage{dsfont}
\usepackage{syntax}
\usepackage{float}
\usepackage{wrapfig}
\usepackage{fancyvrb}
\usepackage{bera}
\usepackage{enumerate}
\usepackage{pxfonts}

\usepackage{pgf}
\usepackage{float}
\usepackage{tikz}
\usetikzlibrary{arrows,automata,shapes, positioning} 
\usepackage{caption}
\usepackage{subcaption}
\makeatletter
\providecommand*{\twoheadrightarrowfill@}{%
  \arrowfill@\relbar\relbar\twoheadrightarrow
}
\providecommand*{\xtwoheadrightarrow}[2][]{%
  \ext@arrow 0579\twoheadrightarrowfill@{#1}{#2}%
}
\makeatother

\lstdefinestyle{pseudo}
{%
backgroundcolor=\color{white},%
frame=lines,%
keywordstyle=\BeraMonottfamily\bfseries,%
morekeywords={spawn,new,def,int,bool,chan,true,false}}

\newcommand{\La}{\ensuremath{\langle}}
\newcommand{\Ra}{\ensuremath{\rangle}}
\newcommand{\tuple}[1]{\La #1 \Ra}
\newcommand{\sem}[1]{[\![ #1 ]\!]}
\newcommand{\relt}[1]{\xrightarrow{#1}}
\newcommand{\reltt}[1]{\xtwoheadrightarrow{#1}}

\newcommand{\abs}[1]{\ensuremath{| #1 |}}

\newcommand{\eval}{\mathrm{eval}}
\newcommand{\At}{\mathrm{Atoms}}

\newcommand{\matches}{\mathrm{matches}}
\newcommand{\map}{\mathrm{map}}
\newcommand{\Post}{\mathrm{Post}}
\newcommand{\shape}{\mathit{shape}}

\newcommand{\sref}[1]{Sec.~\ref{#1}}
\newcommand{\fref}[1]{Fig.~\ref{#1}}
\newtheorem{mydef}{Definition}

\title{Static Analysis of Communicating Processes using Symbolic
  Transducers (Extended version)}

\author{Vincent Botbol$^{1,2}$ \institute{$^1$Sorbonne Universités,
    UPMC Univ Paris 06, CNRS, LIP6 UMR 7606, 4 place Jussieu 75005
    Paris\\$^2$CEA, LIST, Software Reliability and Security
    Laboratory, 91191 Gif-sur-Yvette France\\}
  \email{vincent.botbol@cea.fr} \and\\ Emmanuel Chailloux$^1$
  \email{emmanuel.chailloux@lip6.fr} \and\\ Tristan {Le Gall}$^2$
  \email{tristan.le\_gall@cea.fr} }

\begin{document}

\maketitle
    
\begin{abstract}

We present a general model allowing static analysis based on abstract
interpretation for systems of communicating processes. Our technique,
inspired by Regular Model Checking, represents set of program states
as lattice automata and programs semantics as symbolic
transducers. This model can express dynamic creation/destruction of
processes and communications. Using the abstract interpretation
framework, we are able to provide a sound over-approximation of the
reachability set of the system thus allowing us to prove safety
properties. We implemented this method in a prototype that targets the
MPI library for C programs.

\end{abstract}

This report is an extended version of the research paper submitted to
VMCAI 2017.

\section{Introduction}

The static analysis of concurrent programs faces several well-known
issues, including how to handle dynamical process creation. This last
one is particularly challenging considering that the state space of
the concurrent system may not be known nor bounded statically, which
depends on the number and the type of variables of the program.

In order to overcome this issue, we combine a symbolic representation
based on regular languages (like the one used in Regular Model
Checking~\cite{DBLP:conf/concur/AbdullaJNS04}) with a fixed-point
analysis based on abstract interpretation~\cite{CousotCousot77-1}. We
define the abstract semantics of a concurrent program by using of a
symbolic finite-state transducer~\cite{DBLP:conf/popl/VeanesHLMB12}. A
(classical) finite-state transducer T encodes a set of rules to
rewrite words over a finite alphabet. In a concurrent program, if each
process only has a finite number of states, then we can represent a
set of states of the concurrent program by a language and the
transition function by a transducer. However, this assumption does not
hold since we consider processes with infinite state space, so we have
to represent a set of states of the concurrent program by a lattice
automaton~\cite{DBLP:conf/sas/GallJ07} and its transition function by
a lattice transducer, a new kind of symbolic transducers that we define
in this paper. Lattice Automata are able to recognize languages over an
infinite alphabet. This infinite alphabet is an abstract domain
(intervals, convex polyhedra, etc.) that abstracts process states.

We show, on~\fref{fig:running-example} (detailed
in~\sref{sec:concrete-semantics}), the kind of programs our method is
able to analyse. This program generates an unbounded sequence of
processes $\{id=0, x=5\}; \{id=1, x=9\}; \{id=2, x=13\}; \dots$ We
want to prove safety properties such as: $x = 5 + \text{id} \times 4$
holds for every process when it reached its final location
$l_{\text{~9}}$. The negation of this property is encoded as a lattice
automaton $Bad$~(\fref{fig:safetyproperty}) that recognizes the
language of all bad configurations. Our verification algorithm is to
compute an over-approximation of the reachability set $Reach$, also
represented by a lattice automaton, then, by testing the emptiness of
the intersection of the languages, we are able to prove this property
: $\mathcal{L}(Reach) \cap \mathcal{L}(Bad) = \emptyset$.

\begin{figure}[!h]
  \vspace*{-2em}
  \begin{minipage}[b]{.45\textwidth}
    \begin{lstlisting}%
      [xleftmargin=.2\textwidth, language=C,basicstyle=\ttfamily\scriptsize, numbers=left]
    if (id==0) 
      x := 1
    else
      receive(any_id,x);

    create(next);
    x := x+4;
    send(next,x)
  \end{lstlisting}     
  \captionof{figure}{\label{fig:running-example}Program example}
  \end{minipage}
  \begin{minipage}[b]{.5\textwidth}
    \centering
    \begin{tikzpicture}[->,>=stealth',shorten >=1pt,auto,node distance=2cm,
        semithick]
      \tikzset{initial text={}}
      \tikzstyle{every state}=[]
      
      \node[initial above, state, inner sep=0pt, minimum size=10pt]   (A) {};
      \node[accepting, state, inner sep=0pt, minimum size=10pt] (B) [below of=A]   {};

      \path 
      (A) edge [loop right] node {$\top$} (A)
      
      (A) edge node
      {$ l_{\text{~9}} \times \text{id}\geq 0 \times x \neq 5 + \text{id} * 4$} (B)

      (B) edge [loop right] node {$\top$} (B);

    \end{tikzpicture}
    \captionof{figure}{\label{fig:safetyproperty}Bad configurations}

  \end{minipage}
\end{figure}


\textbf{Related works.} There are many works aiming at the static
analysis of concurrent programs. Some of them use the abstract
interpretation theory, but either they do not allow dynamic process
creation~\cite{mine2014relational} and/or use a different memory
model~\cite{CheckMate09} or do not consider numerical
properties~\cite{DBLP:journals/corr/abs-0802-0188}. In
~\cite{DBLP:conf/popl/VeanesHLMB12}, the
authors defined symbolic transducers but they did not consider to
raise it to the verification of concurrent programs. In
~\cite{DBLP:conf/cav/BouajjaniHV04}, there is
the same kind of representation that considers infinite state system
but can only model finite-state processes. The authors
of~\cite{2009arXiv0910.5833C} present a modular static analysis
framework targeting POSIX threads. Their model allows dynamic thread
creation but lack communications between threads. More practically,
~\cite{vo2009formal} is a formal
verification tool using a dynamical analysis based on model checking aiming at
the detection of deadlocks in Message Passing
Interface~\cite{Snir:1998:MCR:552013} (MPI) programs but this analysis
is not sound and also does not compute the value of the variables.

\textbf{Contributions.} In this article, we define an expressive
concurrency language with communication primitives and dynamic process
creation. We introduce its concrete semantics in terms of symbolic
rewriting rules. Then, we give a way to abstract multi-process program
states as a lattice automaton and also abstract our semantics into a
new kind of symbolic transducer and specific rules. We also give
application algorithms to define a global transition function and
prove their soundness. A fixpoint computation is given to obtain the
reachability set. Finally, in order to validate the approach, we
implemented a prototype as a Frama-C~\cite{Kirchner2015} plug-in which
targets a subset of MPI using the abstract
domain library: Apron~\cite{Jeannet2009}.


\textbf{Outline.} In \sref{sec:concrete-semantics}, we present the
concurrent language and its semantics definition, encoded by rewriting
rules and a symbolic transducer. Then, \sref{sec:abstract-semantics}
presents the abstract semantics and the algorithms used to compute the
over-approximation of the reachability set of a program. In
\sref{sec:mpi}, we detail the implementation of our prototype
targeting a subset of MPI which is mapped by the given semantics and
run it on some examples~(\sref{sec:exp}). We discuss about the
potential and the future works of our method in~\sref{sec:conclusion}.

\section{Programming language and its Concrete Semantics}
\label{sec:concrete-semantics}

We present a small imperative language augmented with communications
primitives such as unicast and multicast
communications, and dynamical process creation. These primitives are
the core of many parallel programming languages and libraries, such
as MPI.

\subsection{Language definition}
\label{sec:language}

In our model, memory is distributed: each process executes the same code, with its
own set of variables. For the sake of clarity, all variables and
expressions have the same type (integer), and we omit the declaration
of the variables. Process identifiers are also integers.
\medskip

\begin{minipage}[t]{.45\textwidth}
\begin{grammar}

<program> ::= <instrs>

<instrs> ::= <instr> ';' <instrs>

<id> ::= <expr> \alt {\bf any\_id}
\end{grammar}
\end{minipage}
\begin{minipage}[t]{.45\textwidth}
\begin{grammar}
<instr> ::= '\{' <instrs> '\}' \alt <ident> ':=' <expr>
\alt {\bf if} '('<expr>')'  <instr> [{\bf else} <instr>]
\alt {\bf while} '('<expr>')' <instr>
\alt {\bf create} '(' <ident> ')'
\alt {\bf send} '(' <id> ',' <ident> ')'
\alt {\bf receive} '(' <id> ',' <ident> ')'
\alt {\bf broadcast} '(' <expr> ',' <ident> ')'
\end{grammar}
\end{minipage}

\begin{grammar}
<ident> and <expr> stand for classical identifiers and arithmetic
expressions on integers (as defined in the C language)
\end{grammar}

Communications are \emph{synchronous}: a process with \texttt{id=orig}
cannot execute the instruction \texttt{send(dest, var)} unless a
process with \texttt{id=dest} is ready to execute the instruction
\texttt{receive(orig, var')}; both processes then execute their
instruction and the value of \texttt{var} (of process \texttt{orig})
is copied to variable \texttt{var'} (of process \texttt{dest}). We
also allow unconditional receptions with \texttt{all\_id} meaning that
a process with \texttt{id=orig} can receive a variable whenever
another process is ready to execute an instruction
$\texttt{send(orig,v)}$. \texttt{broadcast(orig, var)} instructions
cannot be executed unless all processes reach the same
instruction. {\tt create(var)} dynamically creates a new process that
starts its execution at the program entry point. The id of the new
process, which is a fresh id, is stored in {\tt var}, so the current process can communicate
with the newly created process. Other instructions are
asynchronous. Affectations, conditions and loops keep the same meaning
as in the $C$ language.

\subsection{Formal Semantics}

We model our program using an unbounded set $P$ of processes, ordered
by their identifiers ranging from $1$ to $|P|$. 
As usual, the control flow graph (CFG) of the program is a graph where vertices belong to a set $L$ of program points and edges are labelled by a $instr \subseteq L
\times Instr \times L$ where $Instr$ are the instructions defined in
our language. Finally, $V$ represents the set of variables. Their
domain of values is $\mathbb{V} \supseteq \mathbb{N}$. For any expression $expr$
of our language, and any valuation $\rho : V\to \mathbb{V}$, we note
$\eval(expr,\rho) \in \mathbb{V}$ its value.

Our processes share the same code and have distributed
memory: each variable has a local usage in each
process. Thus, a \emph{local state} is defined as $\sigma \in \Sigma =
Id \times L \times (V\to \mathbb{V})$. It records the identifier of
the process, its current location and the value of each
local variable.

A \emph{global state} is defined as a \emph{word} of process local
states: $\sigma_1 \cdot \sigma_2 \cdot ... \cdot \sigma_{n} \in
\Sigma^*$ where $n$ is the number of running processes and $\Sigma^*$
is the free monoid on $\Sigma$. 


The semantics is given as a transition system $\tuple{\Sigma^*, I,
  \tau }$, where $I \in \Sigma^*$ is the set of all possible initial
program states. As the code is shared, every process starts at the
same location $l_0$ and every variable's value is initialised with
0. Therefore, if there are initially $n$ processes, $I = \{\sigma_1
\cdot ... \cdot \sigma_n\}$ where $\forall i \in n, \sigma_i =
\langle i, l_0, (\lambda v~.~0) \rangle$. The transition relation
$\tau \subseteq \Sigma^* \times \Sigma^*$ is defined as:

\begin{itemize}

\item for each local instruction (e.g. assignments, conditionals, and loops)
  $(l,a,l') \in instr$, we have:
  
  $((\sigma_1 \cdot ... \cdot \langle id, l,\rho \rangle \cdot
  ... \cdot \sigma_n ), (\sigma_1 \cdot ... \cdot \langle id, l',\rho'
  \rangle \cdot ... \cdot \sigma_n)) \in \tau\\\text{where } \rho' =
  \sem{a} \rho$ is the classical small-step semantics of action $a$



\item for every pair of {\tt send/receive} instructions of two
  processes :\\ $(l_i, send(id\_to,v_i), l'_i )$ and $(l_j,
  receive(id\_from,v_j), l'_j )$ (or $(l_j,
  receive(\texttt{id\_all},v_j), l'_j )$), we have:

  $ ((\sigma_1 \cdot ... \cdot \langle id_i, l_i,\rho_i \rangle \cdot
  ... \cdot \langle id_j, l_j,\rho_j \rangle \cdot ... \cdot \sigma_n
  ), (\sigma_1 \cdot ... \cdot \langle id_i, l'_i,\rho_i \rangle \cdot
  ... \cdot \langle id_j, l'_j,\rho'_j \rangle \cdot ... \cdot
  \sigma_n )) \in \tau\\ \text{where } \rho'_j = \rho_j [v_j \gets
    \rho_i(v_i)] $ when ($id\_from=\texttt{id\_all}$ or
  $id_i=\eval(id\_from,\rho_j)$) and ($id\_to=\texttt{id\_all}$ or
  $id_j=\eval(id\_to,\rho_i)$ )


\item for each {\tt broadcast} instruction $(l, broadcast(id_x,v),l')$

  $((\langle id_1,l,\rho_1\rangle \cdot ... \cdot \langle id_x,
  l,\rho_x \rangle \cdot ... \cdot \langle id_n,l,\rho_n\rangle),
  (\langle id_1,l',\rho'_1\rangle \cdot ... \cdot \langle id_x,
  l',\rho_x \rangle \cdot ... \cdot \langle id_n,l',\rho'_n\rangle))
  \in \tau\\ \text{where } \forall i\in[1,n], id_i\neq id_x
  \Rightarrow \rho'_i = \rho_i[v\leftarrow \rho_x(v)]$



\item finally, for each {\tt create} instruction $(l, create(v), l')$

  $((\sigma_1 \cdot ... \cdot \langle id, l,\rho \rangle \cdot
  ... \cdot \sigma_n ), (\sigma_1 \cdot ... \cdot \langle id, l',\rho'
  \rangle \cdot ... \cdot \sigma_n \cdot \sigma_{n+1})) \in
  \tau\\\text{where } \rho'= \rho[v \leftarrow n+1] \text{ and
  }\sigma_{n+1} = \langle n+1, l_0, (\lambda v~.~0) \rangle$
\end{itemize}

In the following, we directly consider sets $E \in
\mathcal{P}(\Sigma^*)$ and $\Post_\tau$ defined as:
$$\Post_\tau(E) = \{ w' \in \Sigma^* \, | \,\exists
w \in E \, \wedge \, (w,w') \in \tau \}$$

$\Post_\tau^*$ is the reflexive and transitive closure of
$\Post_\tau$.  Given an initial set of states $I \in \mathcal{P}
(\Sigma^*)$, the \emph{reachability set} $\Post_\tau^*(I)$
contains all states that can be found during an execution of the
program. Assuming we want to check whether the program satisfies a
safety property (expressed as a \emph{bad configuration}) given by a
set of states $B$ that must be avoided, the verification algorithm is
simply to test whether $\Post_{\tau}^*(I) \cap B = \emptyset$; if
true, the program is safe.

Therefore, we would like to define $\Post_{\tau}$ in a more
\emph{operational} way, as a set of rewriting rules that can be
applied to $I$, so we can apply those rules iteratively until we reach
the fixpoint $\Post_{\tau}^*(I)$.

\subsection{Symbolic Rewriting Rules}

Let us consider a local instruction  $(l,a,l') \in instr$; for any set of states $E$:
 $$\Post_{(l,a,l')}(E)=\\\{\sigma_1 \cdot ... \cdot \langle id,
l',\rho' \rangle \cdot ... \cdot \sigma_n ) \; | \; \exists \sigma_1
\cdot ... \cdot \langle id, l,\rho \rangle \cdot ... \cdot \sigma_n)
\in E \; \wedge \; \rho' = \sem{a} \rho \}$$

The effects of $\Post_{(l,al')}$ on $E$ is to rewrite every
word of $E$.  Thus, we would like to express it as a rewriting rule
$G/F$ where $G$ is a symbolic guard matching a set of words and $F$ a
symbolic rewriting function. Since our method uses the framework of
abstract interpretation (see \sref{sec:abstract-semantics}),
\emph{symbolic} means that we consider elements of some lattice to
define the rules. We give the rewriting rule that encodes the
execution of a local instruction $(l,a,l')$:
$$G=\top^* \cdot \langle \_, l, \_ \rangle \cdot \top^* \text{ and }
F= {\tt Id}^* \cdot f \cdot {\tt Id}^* \text{ with } f(X) = \{ \langle
id,l',\sem{a}(\rho) \rangle | \langle id,l,\rho \rangle \in X \} $$
 
The guard matches words composed of any number of processes, then one
process with location $l$, then again any number of processes. The
function ${\tt Id}^*$ means that the processes matched by $\top^*$
will be rewritten as the identity and therefore not modified. $\Lambda
= {\cal P}(\Sigma)$ is the lattice of sets of local states. $f:
\Lambda \to \Lambda$ rewrites a set of local states according to the
semantics of $a$. So every word $w \in E$ that matches the guard will
be rewritten and we will obtain $\Post_{(l,a,l')}(E)$.

We now give the general definition of those rewriting rules and how
to apply them. We remind that the partial order $\sqsubseteq$ can be
extended to $\Lambda^*$ as $u \sqsubseteq v$ if both words have the
same length ($|u|=|v|$) and $\forall i< |u|, \bot \neq u_i \sqsubseteq
v_i$. Note that we do not allow $\bot$ in words: any word that would
contain one or more $\bot$ letters is identified to the smallest element
$\bot_{\Lambda^*}$. Therefore, any word $w\in\Lambda^*$ represents a
set of words of $\Sigma^*$: $\sigma_1\dots \sigma_n \in w$ when
$\{\sigma_1\}\dots \{\sigma_n\} \sqsubseteq w$.

\begin{mydef}
Let $\Lambda$ be a lattice. A rewriting rule over $\Lambda$ is given by two
sequences $G=(g_0)^* \cdot w_1 \cdot (g_1)^* \cdot w_2 \dots w_n \cdot
(g_n)^*$ and $F=f_0 \cdot h_0 \cdot f_1 \cdot h_1 \cdot ... \cdot h_n
\cdot f_{n+1}$ such that:
\begin{itemize}
\vspace*{-1ex}
\item $\forall~1\leq i\leq n, w_i \in \Lambda^*$ and $|w_i|>0$;
\item $\forall~0 \leq i \leq n, g_i \in \Lambda$;\\
We note $N= |w_1| + |w_2| + ... + |w_n|$
\item $\forall~0 \leq i \leq n+1$, $f_i: \Lambda^N\to \Lambda^*$;
\item $\forall~0 \leq i \leq n$, $h_i: \Lambda^{N+1}\to \Lambda$.
\end{itemize}
With this rule, a finite word $w\in \Lambda^*$ is rewritten to $w' \in
\Lambda^*$ if:
\begin{itemize}
\vspace*{-1ex}
\item $w$ can be written as a concatenation $w=u_0 \cdot v_1 \cdot u_1
  \cdot ... \cdot v_n \cdot u_n$ with:
  \begin{itemize}
  \item $\forall~0\leq i\leq n, u_i= \lambda_0 \dots \lambda_{|u_i|}$ and
    $\forall~0 \leq j \leq |u_i|, \lambda_j \sqsubseteq g_i$,
  \item $\forall~1 \leq i \leq n, v_i \sqsubseteq w_i$;
  \end{itemize}
\item $w'=v_0' \cdot u_0' \cdot v_1' \cdot u_1' \cdot ... \cdot v_n'
  \cdot u_n' \cdot v_{n+1}'$ with:
  \begin{itemize}
  \item $\forall~0\leq i\leq n, u_i'= \lambda'_0 \cdot \lambda'_1 \cdot ... \cdot
    \lambda'_{|u_i|}$ and $\forall~0 \leq j \leq |u_i|, \lambda'_j=
    h_i(\lambda_j,v_1,\dots, v_n)$,
  \item $\forall~0 \leq i \leq n+1, v'_i = f_i(v_1,\dots,v_n)$.
  \end{itemize}
\end{itemize}

\end{mydef}

For any $N\in\mathbb{N}$, ${\tt Id}^*: \Lambda^{N+1} \to \Lambda$ is
defined as ${\tt Id}^*(x,y_1,\dots y_N)=x$. Moreover, we denote by
$\langle\_,l,\_\rangle$ the element of $\Lambda = {\cal P}(\Sigma)$
defined as $\{ \langle id,l,\rho \rangle \; | \; \forall id, \rho\} $,
(the symbol '$\_$' matches anything). With these notations, we
can express the transition relation by a set of rewriting rules:

\begin{itemize}
\item For every pair of {\tt send/receive} instructions \\$(l_i,
  {\tt send(id\_to,v_i)}, l'_i )$ and $(l_j, {\tt receive(id\_from,v_j)}, l'_j )$, we
  have the rule:
  $$G=\top^* \cdot \langle \_, l_i, \_ \rangle \cdot
  \top^* \cdot \langle \_, l_j, \_\rangle \cdot \top^* \text{ and }
  F= {\tt Id}^* \cdot f_1 \cdot {\tt Id}^* \cdot f_2 \cdot {\tt
    Id}^*  \text{ with}$$ 
  \begin{itemize}
  \item $f_1(E_1,E_2) = \{\langle id_i,l'_i,\rho_i \rangle \; | \;
    \langle id_i,l_i,\rho_i \rangle \in E_1 \wedge \langle
    id_j,l_j,\rho_j \rangle \in E_2 \wedge id_i = \eval(id\_from,\rho_j) \wedge id_j = \eval(id\_to,\rho_i) \}$
  \item $f_2(E_1,E_2) = \{\langle id_j,l'_j,\rho_j[v_j \gets
    \eval(v_i,\rho_i)] \rangle \; | \; \langle id_i,l_i,\rho_i \rangle \in
    E_1 \wedge \langle id_j,l_j,\rho_j \rangle \in E_2 \wedge id_i = \eval(id\_from,\rho_j) \wedge id_j = \eval(id\_to,\rho_i)\}$
  \end{itemize}
and symmetrically when $\sigma_j$ is located before $\sigma_i$ in the
word of local states. When e.g. $id\_to = {\tt id\_all}$, the
condition $id_j = \eval(id\_to,\rho_i)$ is satisfied for any
$(id_j,\rho_j)$.

\item for each {\tt broadcast} instruction $(l, {\tt broadcast(id\_x,v)},l')$,
  we have the rule:
  $$ G = (\langle \_, l, \_ \rangle)^* \cdot \langle id\_x, l, \_
  \rangle \cdot (\langle \_, l, \_ \rangle)^* \text{ and } F= F_1^*
  \cdot f_1 \cdot F_1^* \text{ with}$$ 
  \begin{itemize}
  \item $F_1^*(E_1, E_2) = \{ \langle id_i, l', \rho_i[v \gets
    \eval(v,\rho_x)] \rangle\; | \; \langle id_i, l, \rho_i \rangle \in E_1
    \wedge \langle id_x, l, \rho_x \rangle \in E_2 \}$
  \item $f_1(E_1) = \{\langle id_x,l',\rho_x \rangle \; | \; \langle
    id_x,l,\rho_x \rangle \in E_1 \}$
  \end{itemize}
The guard $\langle id\_x, l, \_  \rangle $ stands for the set $\{ \langle id_i, l_i, \rho_i  \rangle \, | \, l_i = l \wedge id_i = \eval(id\_x,\rho_i)\}$

  

\item finally, for each {\tt create} instruction $(l, {\tt create(v)}, l')$,
  we have the rule:

  $$ G = \top^* \cdot \{ \_, l, \_\} \cdot \top^* \text{ and } F= Id^*
  \cdot f_1 \cdot Id^* \cdot f_2 \text{ with}$$

  \begin{itemize}
  \item $f_1(E_1) = \{\langle id_i,l',\rho_i \rangle \; | \; \langle
    id_i,l,\rho_i \rangle \in E_1 \}$
  \item $f_2(E_1) = \{ \langle id_n, l_0, (\lambda v~.~0) \rangle\; |
    \; n = {\tt fresh\_id()} \}$
  \end{itemize}

  where {\tt fresh\_id} returns a new unique identifier $n$ where $n =
  |w| + 1$ with $w$ the word of processes.

\end{itemize}

\begin{example}
Let us consider consider our running example depicted
on~\fref{fig:running-example}. Let us assume we have a set of program
states $E=\{\tuple{id=0,l_0,x=0, next=0}; \tuple{id=0,l_9,x=5,
  next=1}.\tuple{id=1,l_8,x=9, next=2}.\tuple{id=2,l_4,x=0,
  next=2}; \tuple{id=1,l_8,x=13, next=2}.\tuple{id=6,l_4,x=0, next=0}.
\}$, i.e. there is either one process in $l_0$, or three process in
$l_9,l_8,l_4$ or two processes in $l_4,l_8$. We consider the symbolic
rewriting rule that results from the communication instructions. Its
guard is $\top^*.\tuple{\_,l_8,\_}.\top^*. \tuple{\_,l_4,\_}.\top^*$ and
its rewriting functions ${\tt Id}^* \cdot f_1 \cdot {\tt Id}^* \cdot
f_2 \cdot {\tt Id}^*$ with 
\begin{itemize}
  \item $f_1(E_1,E_2) = \{\langle id_i,l_9,\rho_i \rangle \; | \;
    \langle id_i,l_8,\rho_i \rangle \in E_1 \wedge \langle
    id_j,l_4,\rho_j \rangle \in E_2 \wedge id_j = \eval(next,\rho_i) \}$
  \item $f_2(E_1,E_2) = \{\langle id_j,l_5,\rho_j[x \gets
    \eval(x,\rho_i)] \rangle \; | \; \langle id_i,l_8,\rho_i \rangle \in
    E_1 \wedge \langle id_j,l_4,\rho_j \rangle \in E_2 \wedge id_j = \eval(next,\rho_i)\}$
  \end{itemize}
then $\Post_\tau(E)= \{ \tuple{id=0,l_9,x=9,
  next=1}.\tuple{id=1,l_9,x=13, next=2}.\tuple{id=2,l_5,x=13, next=2}
\} $, which is the image of the state with three active processes. There is no possible communication when $\tuple{id=1,l_8,x=13, next=2}.\tuple{id=6,l_4,x=0, next=0}$. Even if the locations match the guard, the first process can only send messages to a process with $\texttt{id}= \eval(next,\rho)=2 \neq 6 $. 
\end{example}



\subsubsection{Transducers}

Alternatively, the semantics of local instructions can also be
described by a \emph{lattice transducer}. A finite-state transducer is
a finite-state automaton but instead of only accepting a language, it
also rewrites it. A lattice transducer is similar to a finite-state
transducer; however, it is symbolic, i.e. it accepts inputs (and
produces outputs) belonging to the lattice $\Lambda$, which may be an
infinite set.

\begin{mydef}
A \emph{Lattice Transducer} is a tuple $T = \tuple{\Lambda, Q, Q_0, Q_f, \Delta}$ where:
\begin{itemize}
\item $\Lambda$ is a lattice
\item $Q$ is a finite set of states
\item $Q_0 \subseteq Q$ are the initial states set
\item $Q_f \subseteq Q$ are the final states set
\item $\Delta \subseteq Q \times \Lambda^n \times (\Lambda^n \to \Lambda)^*
  \times Q$  with $n \in \mathds{N}^0$ is a finite set of transitions with guards and rewriting functions
\end{itemize}
\end{mydef}

Let $w=\lambda_1 \cdot ... \cdot \lambda_n \in \Lambda^n$ and $\{q, G, F,
q'\} \in \Delta$ with $G = \gamma_1, ..., \gamma_n$ and ${F = f_1
  , ..., f_m}$. We write $q \xrightarrow{w / w'} q'$ when:
$$
\begin{cases}
  \forall i\in[1,n] \,~\lambda_i \sqsubseteq \gamma_i\\
  
  w' = f_1(\lambda_1, ..., \lambda_n) \cdot ... \cdot
  f_m(\lambda_1, ..., \lambda_n)
\end{cases}
$$

For any word $w\in\Lambda^*$, $T(w)$ is the set of words $w'$ such
that there exists a sequence
$q_0\relt{w_1/w'_1}q_1\relt{w_2/w'_2}\dots \relt{w_n/w_n'}q_f$ with
$q_0\in Q_0$, $q_f\in Q_f$, $w=w_1.w_2\dots w_n$ and
$w'=w'_1.w'_2\dots w'_n$. For any language $L \subseteq \Sigma^*$,
$T(L)= \cup_{w \in L} ~ T(w)$.

We can express the semantics of the local instructions $\langle l_1, a_1,
l_1' \rangle, \langle l_2, a_2,
l_2' \rangle,\dots$ by a transducer as shown in \fref{fig:localtransducer}.

\begin{figure}
  \begin{minipage}[b]{.45\textwidth}
  \centering
  \begin{tikzpicture}[->,>=stealth',shorten >=1pt,auto,node distance=6cm,
      semithick]
    \tikzset{initial text={}}
    \tikzstyle{every state}=[]
    
    \node[initial, state, accepting] (A) {$q_0$};
    
    \path (A) edge [loop above] node { $ \langle \_, l_1, \_ \rangle~/~f
      : (id, l, \rho) \mapsto (id, l_1', \sem{a_1}(\rho))$} (A)

    (A) edge [loop right] node { \dots } (A)
    (A) edge [loop below] node { $ \langle \_, l_2, \_ \rangle~/~f : (id, l,
      \rho) \mapsto (id, l_2', \sem{a_2}(\rho))$} (A);


  \end{tikzpicture}
  \captionof{figure}{\label{fig:localtransducer}Local transitions Transducer}
  \end{minipage}
  \quad 
  \begin{minipage}[b]{.5\textwidth}
    \centering
    \begin{tikzpicture}[->,>=stealth',shorten >=1pt,auto,node distance=2cm,
        semithick]
      \tikzset{initial text={}}
      \tikzstyle{every state}=[]
      
      \node[initial left, accepting, state, inner sep=0pt, minimum size=10pt]   (A) {};

      \path 
      (A) edge [loop above] node {$ \_~/~f : x \mapsto x$} (A)
      
      (A) edge [loop below]  node [right, xshift=5pt, yshift=-10pt, text width=5cm]
      {$\langle \_, l_s, \_ \rangle \cdot \langle \_, l_r, \_
        \rangle~/~$\\$f : (p_s, p_r) \mapsto \langle id_{p_s}, l_s',
        \rho_{p_s} \rangle$\\$f : (p_s, p_r) \mapsto \langle id_{p_r},
        l_r', \rho_{p_r}[x \gets \rho_{p_s}(x)] \rangle$} (A);

    \end{tikzpicture}
    \captionof{figure}{\label{fig:sendright} ``Neighbour''
      communication}

  \end{minipage}
\end{figure}

For the language we presented, the transducer representation is not
fully exploited. Indeed, only single self-looping transitions are
present. Yet, in our example program, we notice that communications
and dynamic creation are done in their ``neighbourhood'': processes
send their $x$ to their right neighbor, receive from the left and
create processes on their right-side. This semantics can be expressed
with our transducer representation. We give on~\fref{fig:sendright} a
transducer encoding a ``neighbour'' version of synchronous
communications as \texttt{send\_right} and \texttt{receive\_left}
primitives. In our illustration, we use the locations $(l_s, l_s')$
and $(l_r, l_r')$ in order to represent pre and post locations of
\texttt{send\_right} and \texttt{receive\_left} instructions. However,
this restriction is not satisfying: we wish to handle point-to-point
communications regardless of process locations in words of
states. Thus we have to limit the transducer to encode only local
transitions.

Therefore, communications are encoded by semantics rules $R$, and
local instructions by a transducer $T$. We note $T_{ext}$ the
transducer extended with semantic rules, i.e. for any language $X
\subseteq \Sigma^*$, $T_{ext}(X) = R(X) \cup T(X) =
\Post_\tau(X)$. For any initial set of states $I \subseteq \{ {\cal
  P}(\Sigma) \}$, we have the reachability set $\Post_\tau^*(I) =
T_{ext}^*(I)$. However, $T_{ext}^*(I)$ cannot be computed in general,
so we need abstractions.

\section{Abstract Semantics}
\label{sec:abstract-semantics}

\subsection{Lattice Automata}
\label{app:lattice-automata}

We give here a look at the lattice automata. The reader may refer to \cite{DBLP:conf/sas/GallJ07} for further details. As said before, the definition of lattice automata requires $\Lambda$ to be atomistic, i.e.:
\begin{itemize}
\item $\At(\Lambda)$ is the set of \emph{atoms}; $\lambda\in\Lambda$
  is an atom if $\forall \lambda'\in\Lambda, \lambda' \sqsubseteq
  \lambda \, \Rightarrow \, \lambda' = \lambda \vee \lambda'= \bot$
\item $\Lambda$ is atomic, i.e. : $\forall \lambda \in \Lambda,\;
  \lambda \neq \bot \Rightarrow \exists \lambda' \in \At(\Lambda),
  \lambda' \sqsubseteq \lambda $
\item any element is equal to to least upper bound of atoms smaller
  than itself: $\forall \lambda \in \Lambda,\; \lambda = \sqcup \{
  \lambda'\,|\,\lambda'\in\At(\Lambda), \lambda' \sqsubseteq
  \lambda\}$
\end{itemize}

The language recognized by lattice automata are on the set of atoms rather than on $\Lambda$ itself. The reason for this is that there may be different edges between the two same nodes. For example, let us consider the lattice of intervals, and let us consider the three automata depicted on~\fref{fig:equivalence-auto}. Intuitively, they represent the same set, but if we define their language as: $w\in \Lambda^*, q_0\reltt{w}q_f$, $\mathcal{L}(A_1)=\{[0,2];[2,4]\}$ while $\mathcal{L}(A_2)=\{[0,3];[3,4]\}$. If we define the language on atoms, both automata recognize the language: $\{[0,0];[1,1];[2,2];[3,3];[4,4]\}$ (assuming we only consider integer bounds). We can also merge transitions and have automaton $A_3$ that recognizes the same language.

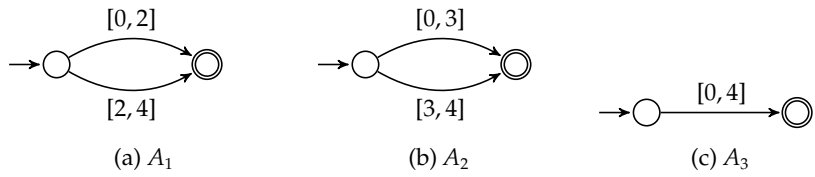
\begin{figure}[ht]
  \begin{subfigure}[b]{.33
\textwidth}
    \begin{tikzpicture}[->,>=stealth',shorten >=1pt,auto,node distance=2cm,
        semithick]
      \tikzset{initial text={}}
      \tikzstyle{every state}=[]
      
      \node[initial,state, inner sep=0pt, minimum size=10pt]   (A)                {};
      \node[accepting, state, inner sep=0pt, minimum size=10pt](B) [right of=A]   {};

      \path[every node/.style={font=\sffamily\small}]
      (A) edge [bend left] node [above] {$[0,2]$} (B);
      \path[every node/.style={font=\sffamily\small}]
      (A) edge [bend right] node [below] {$[2,4]$} (B);
    \end{tikzpicture}
    \caption{$A_1$}
  \end{subfigure}  
  \begin{subfigure}[b]{.3\textwidth}
    \begin{tikzpicture}[->,>=stealth',shorten >=1pt,auto,node distance=2cm,
        semithick]
      \tikzset{initial text={}}
      \tikzstyle{every state}=[]
      
      \node[initial,state, inner sep=0pt, minimum size=10pt]   (A)                {};
      \node[accepting, state, inner sep=0pt, minimum size=10pt](B) [right of=A]   {};

      \path[every node/.style={font=\sffamily\small}]
      (A) edge [bend left] node [above] {$[0,3]$} (B);
      \path[every node/.style={font=\sffamily\small}]
      (A) edge [bend right] node [below] {$[3,4]$} (B);
    \end{tikzpicture}
    \caption{$A_2$}
  \end{subfigure}
  \begin{subfigure}[b]{.3\textwidth}
    \begin{tikzpicture}[->,>=stealth',shorten >=1pt,auto,node distance=2cm,
        semithick]
      \tikzset{initial text={}}
      \tikzstyle{every state}=[]
      
      \node[initial,state, inner sep=0pt, minimum size=10pt]   (A)                {};
      \node[accepting, state, inner sep=0pt, minimum size=10pt](B) [right of=A]   {};

     \path[every node/.style={font=\sffamily\small}]
      (A) edge [] node [above] {$[0,4]$} (B);

    \end{tikzpicture}
    \caption{$A_3$}
  \end{subfigure}
  \caption{\label{fig:equivalence-auto}Three equivalent lattice automata}
\end{figure}

Thus the definition of the language allow us to split or merge
transitions as long as the language remain the same. But if the
interval $[0,+\infty]$ may be split in an infinite number of smaller
intervals, how can we ensure that there is only a finite number of
transitions ? We introduce an arbitrary, finite partition $\pi$ of the atoms. $\pi$ may be defined as a function $\pi: K\to \Lambda$, where $K$ is an arbitrary finite set, such that if $k_1 \neq k_2$, $\pi(k_1)\sqcap\pi(k_2)=\bot$ and $\forall a \in \At(\Lambda), \exists k\in K, a\sqsubseteq \pi(k)$. 

We define \emph{Partitioned Lattice Automata} (PLAs) as the automata such that for any transition $(p,\lambda,q)\in \Delta_A, \exists k \in K, \lambda \sqsubseteq \pi(k)$ (i.e. all the atoms smaller than $\lambda$ belong to the same partition class). A PLA is \emph{merged} if $(p,\lambda,q)\in \Delta_A \wedge (p,\lambda',q)\in \Delta_A \Rightarrow \pi^{-1}(\lambda_1) \neq \pi^{-1}(\lambda_2)$, i.e.  there is at most one transition per element of the partition. So merged PLAs have a finite number of transitions. Moreover, we can use this partition to design algorithms similar to the ones for Finite State Automata (such as union, intersection, determinisation and minimisation), with $K$ playing the role of a finite alphabet. Indeed, if $A$ is a merged PLA, we can apply $\pi^{-1}$ to every label of the transitions and obtain a finite-state automata called $\shape(A)$. \emph{Normalised PLAs} are merged PLAs that are also deterministic and minimised. 

If we have $\nabla_{auto}$, a widening operator on finite-state automata, and $\nabla_\Lambda$ a widening operator on $\Lambda$ then we have a widening operator on lattice automata $A_1 \nabla A_2$:
\begin{itemize}
\item if $\shape(A_1)$ and $\shape(A_2)$ are isomorphic, then we apply $\nabla_\Lambda$ on pairs of isomorphic transitions
\item otherwise we compute $\shape(A_1) \nabla_{auto} \shape(A_2)$ and then merge transitions accordingly.
\end{itemize}

Therefore, lattice automata are a convienient way to ``lift'' a numerical domain $\Lambda$ to an abstract domain for languages over $\At(\Lambda)$, and to extend static analysis of sequential programs to concurrent programs. They can also easily handle disjunctive local invariants: $\lambda_1 \vee \lambda_2$ is simply represented by two transitions $(p,\lambda_1,q)$ and $(p,\lambda_2,q)$. Moreover, the whole reachability set is represented by a single automaton, which is both a blessing and a curse: it provides a concise, graphical way to represent the rechability set, but it also means that when computing a fixpoint by iteration (e.g. computing $T^*(A)$), we compute an increasing sequence of (increasingly large) automata $A_{i+1} = A_i \cup T(A_i)$. When applying $T$ to $A_{i+1}$, we have $T(A_{i+1}) = T(A_i) \cup T(T(A_i))$ should avoid to recompute $T(A_i)$ (either using cache or having a way to apply $T$ only to the `increment').

\subsection{Lattice Automata as an abstract domain}

Since $\Sigma$ may be an infinite set, we must have a way to abstract
languages (i.e. subsets of $\Sigma^*$) over an infinite
alphabet. \emph{Lattice Automata}~\cite{DBLP:conf/sas/GallJ07} provide
this kind of abstractions. Lattice Automata are similar to
finite-state automata, but their transitions are labeled by elements
of a lattice. In our case, lattice automata are appropriate because:
\begin{itemize}
\item they provide a finite representation of languages over an infinite alphabet;
\item we can apply symbolic rewriting rules or a transducer to a lattice automaton (see~\sref{sec:abstract:algo});
\item there is a widening operator that ensures the termination of
  the analysis (see~\sref{sec:abstract:analysis}).
\end{itemize}

\begin{mydef}
A lattice automaton is defined by a tuple $A=\tuple{\Lambda,Q,Q_0,Q_f,\delta}$ where:
\begin{itemize}
\item $\Lambda$ is an atomistic lattice\footnote{See~\cite{DBLP:conf/sas/GallJ07} or Appendix~\ref{app:lattice-automata}.}, the order of which is
  denoted by $\sqsubseteq$;
\item $Q$ is a finite set of states;
\item $Q_0 \subseteq Q$ and $Q_f \subseteq Q$ are the sets of
  initial and final states;
\item $\Delta \subseteq Q \times (\Lambda\setminus\{\bot\}) \times
  Q$ is a finite transition relation.\footnote{No transition is
    labeled by $\bot$.}
\end{itemize}
\end{mydef}

This definition requires $\Lambda$ to have a set of atoms $\At(\Lambda)$. 
Abstract lattices like Intervals~\cite{CousotCousot77-1},
Octagons~\cite{Mine:2006:octagon} and Convex
Polyhedra~\cite{CousotHalbwachs78-POPL} are atomistic, so we can
easily find such lattices to do our static analysis.  Note that if
$\Lambda$ is atomistic, $\Lambda^N$ and $\Lambda^*$ are also
atomistic, their atoms belonging to respectively $\At(\Lambda)^N$ and
$\At(\Lambda)^*$. Moreover, for any set $\Sigma$, the lattice
$\mathcal{P}(\Sigma,\subseteq)$ is atomistic and its atoms are the
singletons. In the remainder of this paper, we will assume that any
lattice we consider is atomistic. Finally, in addition to a widening
operator, lattice automata have classic FSA operations ($\cup$, $\cap$,
$\subseteq$, etc.).

The language recognized by a lattice automaton $A$ is noted
$\mathcal{L}(A)$ and is defined by finite words on the alphabet
$\At(\Lambda)$. $w \in \mathcal{L}(A)$ if $w=\lambda_1 \dots
\lambda_n\in\At(\Lambda)^*$ and there is a sequence of states and
transitions
$q_0\relt{\lambda_0}q_1\relt{\lambda_2}\dots\relt{\lambda_n}q_n$ with
$q_0\in Q_0$ and $q_n\in Q_f$.

The reason why we define the language recognized by a lattice
automaton as sequence of atoms are discussed in
\cite{DBLP:conf/sas/GallJ07}; in a nutshell, this definition implies
that two lattice automata that have the same concretisation recognize
the same language. Moreover, by introducing a finite partition of the
atoms, we can define determinisation and minimisation algorithms
similar to the ones for finite-state automata, as well as a canonical
form (\emph{normalized} lattice automata). 

\paragraph{Abstractions and Concretisations} Assuming there is a Galois connection between $\mathcal{P}(\Sigma)$ and $\Lambda$ we can extend the  concretisation function $\gamma: \Lambda \to \mathcal{P}(\Sigma)$, we can extend it to $\gamma: \Lambda^* \to \mathcal{P}(\Sigma^*)$; if $w=\lambda_1\dots \lambda_n \in \Lambda^*$, $\gamma(w) = \{ \sigma_1\dots\sigma_n | \forall i=1..n,  \sigma_i \in \gamma(\lambda_i)\}$ and for any language $L$, $\gamma(L) = \cup_{w \in L} \gamma(w)$. Thus, the concretisation of a lattice automaton $A$ is $\gamma(\mathcal{L}(A))$, which can be computed by applying $\gamma$ to all  of $A$.  Lattice automata are not a complete lattice; the abstraction function is defined as: if $L$ is regular (i.e. it can be represented by a lattice automaton with labels in $\mathcal{P}(\Sigma)$) $\alpha(L)$ we can apply $\alpha$ to each edge; otherwise $\alpha(L)=\top$. The latter case does not happen in practice, since the initial set of states $I$ is regular, and since we only check regular properties.
We now present algorithms to apply a symbolic rewriting rule or a
lattice transducer to a lattice automaton.


\subsection{Algorithms}
\label{sec:abstract:algo}

\subsubsection{Application of a Rule}
To apply a symbolic rewriting rule to the language recognized by a
lattice automaton, we must first identify the subset of words that
match the guard $(g_0)^* \cdot w_1 \cdot (g_1)^* \cdot w_2 \dots w_n
\cdot (g_n)^*$. In this guard, it's easier to look first for sequences
in the automaton that match $w_1,w_2,\dots,w_n$. In automaton $A$ a
sequence that matches e.g. $w_1$ begins from state $q_b^1$ and ends in
state $q_e^1$. Then, we identify the sub-automaton that could match
$(g_0)^*$, i.e. all the states that are reachable from an initial
state $q_0$ and correachable from $q_b^1$ by considering only
transitions labeled by elements $\lambda$ such that $g_0 \sqcap
\lambda \neq \bot$. Once each part is identified, we can apply the
rewriting function to each part and then we get a new automaton
$A'$. Since this pattern matching is non deterministic, we have to
consider all possible matching sequences. The result of the algorithm
is the union of every automaton $A'$ constructed in this way.

We introduce some notations before writing the algorithm. Let
$w=\lambda_1\dots \lambda_n \in \Lambda^n$ and let $A$ be a lattice
automaton. We denote by $\matches(w,A)$ the set of matching
sequences:\begin{align*}
  \matches(w,A) = \{ &(q_b,v_1\dots v_n ,q_e)~|~\exists q_0
  \relt{\lambda_1'} q_1 \relt{\lambda_2'} \dots \relt{\lambda_n'} q_n
  \in A,& \\&q_0=q_b \wedge q_n = q_e \wedge \forall i= 1..n, v_i =
  \lambda_i \sqcap \lambda'_i \neq \bot \} \} &
\end{align*}
Let $(q_b,q_e)$ be a
pair of states of a lattice automaton
$A=\tuple{\Lambda,Q,Q_0,Q_f,\delta}$ and let $\lambda\in\Lambda$. We
denote by $A_{q_b \to q_e}$ the sub-automaton $A_{q_b \to
  q_e}=\tuple{\Lambda,Q,\{q_b\},\{q_e\},\delta}$. For a lattice
automaton $A=\tuple{\Lambda,Q,Q_0,Q_f,\delta}$ and a function $f:
\Lambda \to \Lambda$, we denote by $\map(f,A)$ the automaton
$\map(f,A)= \tuple{\Lambda,Q,Q_0,Q_f,f(\delta)}$ where $f(\delta) = \{
(q,f(\lambda),q') | (q,\lambda,q') \in \delta \wedge f(\lambda) \neq
\bot\}$.

With those notations, we give an algorithm to apply a rewriting rule
on a lattice automaton:

\begin{lstlisting}[style=pseudo,mathescape=true]
ApplyRule($G = (g_0)^* \cdot w_1 \cdot (g_1)^* \dots w_n \cdot
(g_n)^*$ and $F =f_0 \cdot h_0 \cdot f_1 \cdot ... \cdot h_n
\cdot f_{n+1},~ A$):
Result $:= \emptyset$ 
For all matching sequences 
$(q_b^1,v^1, q_e^1) \in \matches(w_1,A)$, $\dots$, $(q_b^n,v^n, q_e^n) \in \matches(w_n,A)$, 
for each initial state $q_0 \in Q_0^A$ and for each final state $q_f \in Q_f^A$
$($ Let $A_0=\map(x\mapsto g_0 \sqcap x,A)_{q_0 \to q_b^1}$,
  $A_1=\map(x\mapsto g_1 \sqcap x,A)_{q_e^1 \to q_b^2}$, $\dots$, $A_n=\map(x\mapsto g_n \sqcap x,A)_{q_e^n \to q_f}$. 
  For $i=0~..~n$:
    let $A'_i=\map(x \mapsto h_i(x,v^1,\dots,v^n),A_i)$. 
  For $i=0~..~n+1$:
    let $w'_i=f_i(v^1,\dots,v^n)$.
  Let $q_{-1}$ and $q_{n+1}$ be two fresh states (not appearing in any $A'_i$).  
  Let $\delta^{seq} = \{ (q_{-1},w'_0,q_0)
  (q_b^1,w'_1,q_e^1) (q_b^1,w'_1,q_e^1) \dots (q_b^n,w'_n,q_e^n)
  (q_b^n,w'_{n+1},q_{n+1})\}$.
  Let $A'= \tuple{\Lambda,Q \cup \{q_{-1},q_{n+1}\},\{q_{-1}\},\{q_{n+1}\},\delta^{A'}}$ with
  $\delta^{A'}= \delta^{seq} \cup \delta^{A'_0} \cup \dots \cup
  \delta^{A'_n}$ 

  Result $:=$ Result $\cup~A'~)$
return Result
\end{lstlisting}

\begin{theorem}
  \label{theorem:ruleinclusion}
Let $R=(g_0)^* \cdot w_1 \cdot (g_1)^* \cdot w_2 \dots w_n \cdot
(g_n)^*\;/\;f_0 \cdot h_0 \cdot f_1 \cdot h_1 \cdot ... \cdot h_n
\cdot f_{n+1}$ be a rewriting rule and $A$ a lattice automaton. If $R(A)= ApplyRule(R,A)$, then we have: $R(\mathcal{L}(A)) \subseteq \mathcal{L}(R(A))$. 

\end{theorem}


\begin{proof}
Let $R= (g_0)^* \cdot
w_1 \cdot (g_1)^* \cdot w_2 \dots w_n \cdot (g_n)^* / f_0 \cdot h_0 \cdot f_1 \cdot ... \cdot h_n \cdot f_{n+1}$ 


Let $w \in \mathcal{L}(A)$. If $w$ matches the guard $(g_0)^* \cdot
w_1 \cdot (g_1)^* \cdot w_2 \dots w_n \cdot (g_n)^*$, it means we can
decompose it as $w=u_0.v_1.u_1.v_2.u_2 \dots v_n.u_{n}$ such that:
\begin{itemize}
\item for $i=1..n$, $v_i \sqsubseteq w_i$
\item for $i=0..n$, each letter of $u_i$ is smaller than $g_i$
\end{itemize}
Since $w \in \mathcal{L}(A)$, we consider a path
$q_0\reltt{u_0}q_b^1\reltt{v_1} q_e^1\dots \reltt{v_n}q_e^n\reltt{u_n}
q_f$ in $A$, i.e. there are matching sequences $(q_b^1,v^1, q_e^1) \in
\matches(w_1,A)$, $(q_b^2,v^2, q_e^2) \in \matches(w_2,A)$, $\dots$,
$(q_b^n,v^n, q_e^n) \in \matches(w_n,A)$, such that $\forall i, v_i
\sqsubseteq v^i$. In algorithm  ${\tt ApplyRule}$, these matching sequences generate an automaton $A'$. By applying the rewriting functions $f_0 \cdot h_0 \cdot f_1 \cdot h_1 \cdot ... \cdot h_n
\cdot f_{n+1}$ to $w$, we obtain $R(w)$ which is recognized by $A'$. So $R(w)\in \mathcal{L}(A') \subseteq \mathcal{L}(R(A))$.

\qed
\end{proof}

However, we do
not have $R(\mathcal{L}(A)) \supseteq \mathcal{L}(R(A))$ as shown in
the following example:
\begin{example}
Let $A$ the lattice automaton that recognizes the language\\
$L = \{[0,0],[1,1]\}$ (i.e. there is only one process, with one integer
variable which values is either $0$ or $1$) and let $R$ be the
rewriting rule $\top/f_1.f_2$ where $f_1: x\mapsto 2x $ and $f_2: x
\mapsto 4x$ then $R(\mathcal{L}(A))=
\{[0,0].[0,0]\,,\,[2,2].[4,4]\}$. But $R(A)$ is a lattice automaton
that can recognize 4 words:\\ $\mathcal{L}(R(A))=
\{[0,0].[0,0]\,,\,[2,2].[4,4]\,,\,[0,0].[4,4]\,,\,[2,2].[0,0]\}$.
\end{example}

\subsubsection{Application of a transducer}


The following algorithm computes the application of a symbolic
transducer $T=\tuple{\Lambda,Q^T,Q_0^T,Q_f^T,\Delta^T}$ to a (language
recognized by a) lattice automaton
$A=\tuple{\Lambda,Q^A,Q_0^A,Q_f^A,\Delta^A}$. The idea is to consider
the cartesian product $Q^T \times Q^A$ and to create transitions
whenever it is allowed by the transducer and the automaton.

\begin{lstlisting}[style=pseudo,mathescape=true]
ApplyTransducer(T,A):

$\Delta^{T(A)}= \emptyset$
$\forall (p,q) \in Q^{T} \times Q^{A}$ 
$\forall p' \in Q^{T}, \forall (p,Q,F,p') \in \Delta^T$ with $G \in \Lambda^n$ and $F: \Lambda^n\to\Lambda^*$ 
$\forall q' \in Q^{A}$ such that there is a sequence of transitions
  $q \reltt{w} q'$ in $A$ with $w\in\Lambda^n$
If $G \sqcap w \neq \bot$ then $\Delta^{T(A)}:= \Delta^{T(A)}\cup
\left\{( (p,q), F(G \sqcap w), (p',q') )\right\}$
$\forall(p,q) \in Q^{T(A)}_0$ if $p \in Q^T_0$ and $q \in Q^A_0$
$\forall(p,q) \in Q^{T(A)}_f$ if $p \in Q^T_f$ and $q \in Q^A_f$
\end{lstlisting}

Note that $\Delta^{T(A)}= \Delta^{T(A)}\cup \left\{( (p,q), F(G \sqcap
w), (p',q') )\right\}$ means that we add not one but a sequence of
transitions (introducing fresh new states). So the set of states of
the resulting automata $T(A)$ is the union of $Q^T \times Q^A$ and all
the fresh states we added. \fref{fig:applicationexample} gives an
illustration of an application of a transducer mapping the semantics
of the single local instruction of our program $(l_7, \texttt{[x := x
    + 4]}, l_8)$ on single letter program state set (i.e. only one
process). Please note that, for the sake of clarity, we use line
numbers as locations. $l_7$ is the location just before the evaluation
of the assignment. $l_8$ is, thus, after the evaluation and $l_9$
represents the last location symbolising the end of a process
execution.  Our transducer application algorithm complexity is
$O(\abs{Q_A} \cdot \abs{Q_T} \cdot \abs{\Delta_T} \cdot \abs{\pi}^N)$
where $\pi$ is the lattice automata's partition (its size depends on
the locations of the program) and $N$ is the maximum length of all
transition guards (here $N=1$).

\begin{figure}[ht]

  \begin{subfigure}[b]{.4\textwidth}
    \begin{tikzpicture}[->,>=stealth',shorten >=1pt,auto,node distance=4cm,
        semithick]
      \tikzset{initial text={}}
      \tikzstyle{every state}=[]
      \node[initial,state,inner sep=0pt, minimum size=10pt, accepting]   (A)  {};
      \path (A) edge [loop above] node {$\langle \_, l_7, \_ \rangle
        ~/~f : (id,l,\rho) \mapsto id, l_8,\rho[x \gets \rho(x) + 4]$} (B);
    \end{tikzpicture}
    \caption{$T$}
  \end{subfigure} \qquad
  \begin{subfigure}[b]{.25\textwidth}
    \begin{tikzpicture}[->,>=stealth',shorten >=1pt,auto,node distance=2cm,
        semithick]
      \tikzset{initial text={}}
      \tikzstyle{every state}=[]
      
      \node[initial,state, inner sep=0pt, minimum size=10pt]   (A)                {};
      \node[accepting, state, inner sep=0pt, minimum size=10pt](B) [right of=A]   {};

      \path 
      (A) edge node [yshift=5pt] {$(id_0, l_7, \{ x \mapsto 1\})$} (B);
    \end{tikzpicture}
    \caption{A}
  \end{subfigure}  
  \begin{subfigure}[b]{.25\textwidth}
    \begin{tikzpicture}[->,>=stealth',shorten >=1pt,auto,node distance=2cm,
        semithick]
      \tikzset{initial text={}}
      \tikzstyle{every state}=[]
      
      \node[initial,state, inner sep=0pt, minimum size=10pt]   (A)                {};
      \node[accepting, state, inner sep=0pt, minimum size=10pt](B) [right of=A]   {};

      \path 
      (A) edge node [yshift=5pt] {$(id_0, l_8, \{ x \mapsto 5\})$} (B);
    \end{tikzpicture}
    \caption{$T(A)$}
  \end{subfigure}
  \caption{\label{fig:applicationexample}Transducer application}
  
\end{figure}
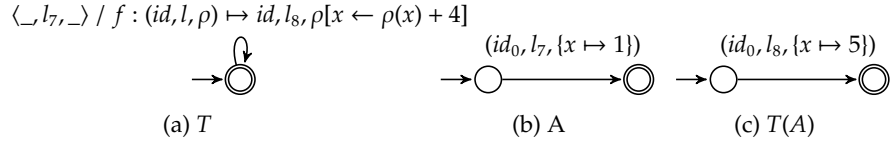

\begin{theorem}
  \label{theorem:transducerinclusion}
Let $T$ be a symbolic transducer and $A$ a lattice automaton. We have:
$T(\mathcal{L}(A)) \subseteq \mathcal{L}(T(A))$.  
\end{theorem}

\begin{proof}
  Let $T=\tuple{\Lambda,Q^T,Q_0^T,Q_f^T,\Delta^T}$ and $A=\tuple{\Lambda,Q^A,Q_0^A,Q_f^A,\Delta^A}$.
  Let $w \in \mathcal{L}(A)$; we must prove that $T(w) \subseteq \mathcal{L}(T(A))$.
  By definition, $w'\in T(w)$ if $p_0\reltt{w_1/w'_1}p_1\reltt{w_2/w'2}\dots \reltt{w_n/w'_n}p_n$
  with $p_0\in Q_0^T,p_n\in Q_f^T$, $w=w_1.w_2\dots w_n$, $w'=w'_1.w'_2\dots w'_n$ 
  and $\forall i=1..n, \exists (p_{i-1},G_i,F_i,p_i)\in\Delta^T$ with
  $w_i \sqsubseteq G_i$ and $w'_i=F(w_i)$.

  Since $w\in\mathcal{L}(A)$, it means $\exists q_0...q_n \in Q^A$
  with $q_0 \in Q^A_0, q_n \in Q^A_f$ such that, $q_0 \reltt{w_1}
  q_1\reltt{w_2} \dots \reltt{w_n} q_n$. In other words $\forall
  i=1..n$ there is a sequence of transitions:
  $q_{i-1}\relt{\lambda_{i,1}}\dots\relt{\lambda_{i,m_i}}q_i$ with
  $w_i \sqsubseteq \lambda{i,1}\dots\lambda_{i,m_i}$.  So $w_i
  \subseteq G_i \sqcap \lambda{i,1}\dots\lambda_{i,m_i}$ and
  $F(w_i)=w'_i\sqsubseteq F(G_i \sqcap
  \lambda{i,1}\dots\lambda_{i,m_i})$. By definition of $T(A)$,
  $(p_0,q_0)\reltt{w'_1}\dots\reltt{w'_n}(p_n,q_n)$ is
  an accepting run of $T(A)$, thus $w'\in\mathcal{L}(T(A))$.

The same principle applies for R(A).

\qed \end{proof}

We note $T_{ext}(A) = R(A) \cup T(A)$ the automaton resulting of the
union of ${\tt ApplyRule(R,A)}$ and ${\tt ApplyTransducer(T,A)}$.

\subsection{Fixpoint computation}
\label{sec:abstract:analysis}

As we said before, the reachability set is defined as the fixpoint
$\Post_\tau^*(I)$; If we can compute $T_{ext}^*(I)$ in the
abstract domain of lattice automata, we will get an over-approximation
interpretation of this reachability set. However, there are infinitely
increasing sequences in this abstract domain, so we need to apply a
\emph{widening operator} to ensure the termination of the
computation. There exists a widening operator which ``lifts'' a
widening operator $\nabla_\Lambda$ defined for $\Lambda$ to the
abstract domain of lattice automata: $A_1 \nabla A_2$ applies
$\nabla_\Lambda$ to each transition of $A_1$ and $A_2$ when the two
automata have the same ``shape''; otherwise, it merges some states of
$A_1 \cup A_2$ to obtain an over-approximation
(see~\cite{DBLP:conf/sas/GallJ07}).

The generic fixpoint algorithm is thus to apply the widening operator $\nabla$ at each step until we reach a post-fixpoint, i.e. we iterate the operator
  \begin{equation*}
    T_\nabla(S) = 
    \begin{cases}
      S & \text{if } T_{ext}(S) \sqsubseteq S\\
      S \nabla (S \cup T_{ext}(S)) &\text{otherwise }
    \end{cases}
  \end{equation*}

This computation gives a post-fixpoint $T^\infty \supseteq
T_{ext}(I)$. In practice, this method may yield very
imprecise upper bounds. Since $\Lambda$ contains information about the
location of each process, we can improve the precision by applying
$\nabla_\Lambda$ only to locations corresponding to an entry point of
a loop. It is known
~\cite{Bourdoncle93efficientchaotic} that we only need
widening to break dependency cycles
and~\cite{Bourdoncle93efficientchaotic} gives an extensive study on
the choice of widening application locations.

Once we get an over-approximation of the reachability set, we can check
any safety property expressed as a set of bad states represented by a
lattice automaton $B$; if $T^\infty \cap B = \emptyset$, then the
system is safe. If not, the property may be false, thus we raise an
alarm.

On our example (\fref{fig:running-example}), applying our method using
a precise relational numerical abstract domain (e.g. polyhedra) gives
us a reachability set. We can prove the safety property given
on~\fref{fig:safetyproperty} by using the following
invariant present in the reachability set:

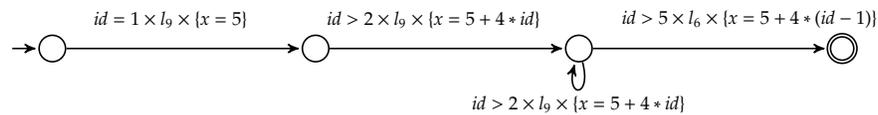
\begin{center}
  \scriptsize
  \begin{tikzpicture}[->,>=stealth',shorten >=1pt,auto,node distance=3.5cm,
    semithick] 
  \tikzset{initial text={}} 
  \tikzstyle{every state}=[]
  
  \node[initial, state, inner sep=0pt, minimum size=10pt]   (A) {};
  \node[state, inner sep=0pt, minimum size=10pt] (B) [right of=A]   {};
  \node[state, inner sep=0pt, minimum size=10pt] (C) [right of=B]   {};
  \node[accepting, state, inner sep=0pt, minimum size=10pt] (D) [right of=C]   {};

  \path 
  (A) edge node [xshift=-5pt, yshift=5pt] {$id = 1 \times l_9 \times \{ x = 5 \}$} (B)
  
  (B) edge node [xshift=-5pt, yshift=5pt] {$id > 2  \times l_9 \times \{ x = 5 + 4 * id \}$} (C)

  (C) edge [loop below] node {$id > 2  \times l_9 \times \{ x = 5 + 4 * id \}$} (C)

  (C) edge node [xshift=15pt, yshift=5pt] {$id > 5 \times l_6 \times \{
    x = 5 + 4 * (id - 1) \}$} (D);
\end{tikzpicture}
\end{center}

\section{Verification of MPI programs}
\label{sec:mpi}

In order to validate our approach, we applied our method to the
Message Passing Interface (MPI). MPI is a specification of a message
passing model. Many implementations have been developed and it is
widely used in parallel computing for designing distributed
programs. Every process has its own memory and \emph{shares a common
  code}. A notion of \emph{rank} (acting as id) is present in order to
differentiate the processes. This paradigm makes a good candidate to
map our model onto.

We developed a prototype\footnote{The prototype can be found at:
  \url{https://www-apr.lip6.fr/~botbol/mpai}} that targets a MPI
subset for the C language. It currently supports \emph{synchronous MPI
  communications}, integer and floating point values as well as a good
subset of the C language. Currently, we do not support dynamic process
creation in MPI.
This prototype has been implemented as a \emph{Frama-C}
plug-in.
This plug-in uses a lattice transducer library we developed on top of
an existing lattice automata implementation. Our abstract domains are
given by the \texttt{Apron} library. This prototype has been written
in OCaml. The current size of the plug-in is around 10.000 lines of code and is
still a work in progress. Unfortunately, due to licensing issues, its
source code is not available yet.



To illustrate our method, we refer, throughout this section, to a
small MPI program (\fref{fig:simpleprog}). This program runs $N$
processes that each computes $1/2^{(rank+1)}$. Then, the root
(i.e. rank = 0) process collects each local result and sums them by a
call to the \texttt{MPI_Reduce} primitive.

\begin{figure}[ht]
  \begin{lstlisting}[xleftmargin=.1\textwidth, language=C,basicstyle=\ttfamily\scriptsize,keywordstyle=\color{blue}, numbers=left ]
    int main(int argc, char **argv) {
      int rank, i;
      float res, total;
      MPI_Init(&argc, &argv);
      MPI_Comm_rank(MPI_COMM_WORLD, &rank);
      i = 1 << (rank + 1);
      res = 1. / i;
      MPI_Reduce(&res, &total, 1, MPI_FLOAT, MPI_SUM, 0, MPI_COMM_WORLD);
      MPI_Finalize();
      return 0;
    }
  \end{lstlisting}   
  
  \caption{\label{fig:simpleprog}MPI program computing: $\sum\limits_{i=1}^n
    \frac{1}{2^i}$}
\end{figure}

\subsection{Program state representation}

Each (abstract) local process state is a tuple $\tuple{l,\lambda} \in L\times \Lambda$, where $L$ is the set of locations and $\Lambda$ a numerical abstract lattice. In the examples of this section, $\Lambda$  is the lattice of Intervals. Moreover, we distinguish the value of $Id$ from the other variables.  


To illustrate, we give the initial configuration with $2$ processes
starting at \texttt{MPI_Init(\&argc, \&argv)} (variable
declarations are omitted) and represented as a lattice automaton. At
this point, each environment variable is set to $\top$ meaning they
are not initialized and can have any possible value.
\medskip

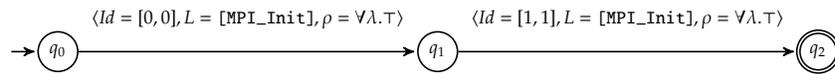
\begin{tikzpicture}[->,>=stealth',shorten >=1pt,auto,node distance=4cm,
    semithick]
  \scriptsize
  \centering
  
  \tikzset{initial text={}}
    \tikzstyle{every state}=[]
    
    \node[initial left, state, minimum size=15pt] (A) {$q_0$}; \node[state,
      inner sep=0pt, minimum size=15pt] (B) [right=4.5cm of A] {$q_1$};
    \node[accepting, state, inner sep=0pt, minimum size=15pt] (C) [right=4.5cm of
      B] {$q_2$};

    \path (A) edge node [yshift=6pt] {$\langle Id = [0,0], L =
      \texttt{[MPI_Init]}, \rho = \forall \lambda. \top\rangle $} (B)
    (B) edge node [yshift=6pt] {$\langle Id = [1,1], L =
      \texttt{[MPI_Init]}, \rho = \forall \lambda. \top\rangle $} (C);
  \end{tikzpicture}

\subsection{Transducer automatic generation}

Starting from a MPI/C program, the goal is to automatically generate a
lattice transducer that fully encodes the program semantics. To
achieve that, we first compute the program's Control Flow Graph
(CFG). Then, we translate each CFG transition into a lattice
transducer rule yielding the complete transducer encoding the program
semantics.


As stated before, we differentiate local instructions that affects
only one process at a time from global instructions, such as MPI
communications, that modify the global state of the program. The
translation of local instructions is straight-forward: we use
classical transfer functions that are defined in the Apron
library to evaluate the expressions and do the assignments. As shown
below, an ``if'' C statement will be translated into
two corresponding rules for both condition cases.

  \begin{minipage}[b]{.45\textwidth}
    \begin{lstlisting}[xleftmargin=.2\textwidth,language=C,basicstyle=\ttfamily\scriptsize,keywordstyle=\color{blue}, numbers=left ]
if (x > 10){
   ...  
} else { 
   ...  
} 
      ...
    \end{lstlisting}   
    \captionof*{subfigure}{C program}
  \end{minipage}
  \begin{minipage}[b]{.45\textwidth}
    \centering
    \scriptsize
    \begin{tikzpicture}[->,>=stealth',shorten >=1pt,auto,node distance=4cm,
        semithick]
      \tikzset{initial text={}}
      \tikzstyle{every state}=[]
      
      \node[initial above, state, inner sep=0pt, minimum size=10pt]
      (A) {If}; 
      
      \node[state, inner sep=0pt, minimum size=10pt](B) [below
        left=0.8cm of A] {$\text{L}_2$};

      \node[state, inner sep=0pt, minimum size=10pt](C) [below right=0.8cm
        of A] {$\text{L}_4$};

      \node[state, inner sep=0pt, minimum size=10pt](D) [below=1.3cm of
        A] {$\text{L}_6$};

      \path (A) edge [above left] node {$x > 10$} (B)
      (A) edge node {$x \leq 10$} (C)
      (B) edge [below left] node {$\dots$} (D)
      (C) edge node {$\dots$} (D);
    \end{tikzpicture}
    \captionof*{subfigure}{CFG}
  \end{minipage}

  \begin{minipage}[b]{\textwidth}
    \centering
    \footnotesize
    \begin{tikzpicture}[->,>=stealth',shorten >=1pt,auto,node distance=4cm,
        semithick] \tikzset{initial text={}} \tikzstyle{every state}=[]
      \node[initial, accepting,state,inner sep=0pt, minimum size=10pt]
      (A) {};

      \path (A) edge [loop above] node {$\top \times \texttt{[If]} \times
        \{ x \in [11,+\infty]\}~/~f : (id,L,\rho) \mapsto id,{\tt [L_2]}, \rho$}
      (A);
      \path (A) edge [out=260, in=290, looseness=12] node [below] {$\top \times \texttt{[If]} \times
        \{ x \in [-\infty, 10]\}~/~f : (id,L,\rho) \mapsto id,{\tt [L_4]}, \rho$}
      (A);
      \path (A) edge [loop right] node {$\dots$} (A);
    \end{tikzpicture}      
    \captionof*{subfigure}{Resulting transducer}
  \end{minipage}

Below is the transducer generated from all local instructions of the
MPI program depicted on~\fref{fig:simpleprog}. Note that, with this
set of local rules, there is no way to evolve from the {\tt
  MPI_Reduce} location. As mentioned in the previous section, we
dissociate the global rules from the transducer's local
rules. Therefore, this transition will be presented in the next
section. Finally, in order to model process inactivity, we add a
simple rule $\top~/~f : x \mapsto x$ meaning that any process at any
location might not evolve.

\begin{tikzpicture}[->,>=stealth',shorten >=1pt,auto,node distance=4cm,
    semithick] \tikzset{initial text={}} \tikzstyle{every state}=[]
  \scriptsize
  \centering
  \node[initial, accepting,state,inner sep=0pt, minimum size=10pt] (A) {};


  \path (A) edge [loop above, text width=12cm] node {\begin{align*}\top \times \texttt{[MPI_Init]} \times
      \top~&/~f : (id, l, \rho) \mapsto id, \texttt{[MPI_Comm_rank]}, \rho\\
      \top \times \texttt{[MPI_Comm_rank]} \times
      \top~&/~f : (id, l, \rho) \mapsto id, \texttt{[i = 1 <{}< (rank + 1)]}, \rho[rank \gets id]\\      
      \top \times \texttt{[i = 1 <{}< (rank + 1)]} \times
      \top~&/~f : (id, l, \rho) \mapsto id, \texttt{[res = 1. / i]}, \rho[i \gets 1 \text{<{}<} (\rho(rank) + 1)]
    \end{align*}
  }
  (A);

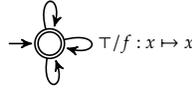
  \path (A) edge [out=260, in=290, looseness=12, text width=12cm] node [below] {\begin{align*}
      \top \times \texttt{[res = 1. / i]} \times
      \top~&/~f : (id, l, \rho) \mapsto id, \texttt{[MPI_Reduce]}, \rho[res \gets 1. / \rho(i)]\\
      \top \times \texttt{[MPI_Finalize]} \times
      \top~&/~f : (id, l, \rho) \mapsto id, \texttt{[return]}, \rho
  \end{align*}}
  (A);
  \path (A) edge [loop right] node {$\top / f : x \mapsto x$} (A);
  
\end{tikzpicture}      

\subsection{Encoding communication primitives}

Our prototype currently accepts this subset of MPI primitives :
\texttt{MPI_Send}, \texttt{MPI_Recv}, \texttt{MPI_Bcast},
\texttt{MPI_Comm_rank}, \texttt{MPI_Comm_size} and
\texttt{MPI_Reduce}. We already described the symbolic rewriting rules
in \sref{sec:abstract-semantics} except for \texttt{MPI_Comm_rank},
\texttt{MPI_Comm_size}, which returns the id of the current process
and the total number of processes, and \texttt{MPI_Reduce}. Let us
give the semantics of the last one:

\begin{Verbatim}[fontsize=\small]
MPI_Reduce(void* send_data, void* recv_data, int count, 
      MPI_Datatype datatype, MPI_Op op, int root, MPI_Comm communicator)
\end{Verbatim}

This global communication primitive gathers every process'
\texttt{send_data} buffer and applies a commutative (the order of
reduction is undefined) operator \texttt{op} between every value. The
result is then sent to the process of rank \texttt{root} at its
\texttt{recv_data} address. \texttt{count} and \texttt{datatype} are
respectively the size of these buffers and the type of each value. The
\texttt{communicator} defines a group of processes where the
communication will occur. We assume a single group.

We cannot represent this global communication in our model with only
one rule. Our solution is to break it down into three different
ones. The main idea is to spawn a ``collector'' process that will be
in charge of gathering each process' \texttt{send_data} and applying
the reduction operation on its accumulator. This collector will move
through the program state (i.e. a word of local states) by swapping,
at each iteration, with the next process. Before starting to move this
collector, we have to ensure that no involved process might
evolve. Therefore, we \emph{lock} them using a special location.  When
the collector reaches the end of the word, it sends its accumulator to
the root process through a point-to-point communication, destroys
itself and, finally, unlocks the processes. The three skeleton rules
used in our prototype are given here:
\begin{enumerate}
\item \parbox[t]{\textwidth}{
   \vspace{-2.3em} \begin{flalign*}
  G_1 &= ( \top \times \texttt{[MPI_Reduce]} \times \top ) ^*&\\
  F_1 &= (f : \_ \mapsto -1, \texttt{[Collector]}, \{ recv\_data \mapsto e \})~\cdot&\\
  &(F : ((id,l,\rho), \_) \mapsto id, \text{lock}(l), \rho) \text{ where } e \text{ is the neutral element of \texttt{op}}&\\
\end{flalign*}}
  \vspace{-2.5em}
  \item \parbox[t]{\textwidth}{
   \vspace{-2.3em} \begin{flalign*}
    G_2 &= ( \top \times \texttt{[MPI_Reduce]}_{\text{lock}} \times \top )^*~\cdot&\\
    &(\top \times \texttt{[Collector]} \times \top)~\cdot~( \top \times \texttt{[MPI_Reduce]}_{\text{lock}} \times \top )~\cdot&\\
    &( \top \times \texttt{[MPI_Reduce]}_{\text{lock}} \times \top )^*\\
    F_2 &= {\tt Id}^*~\cdot&\\
    &(f : ((id_{\text coll},l_{\text coll},\rho_{\text coll}), (id_{\text proc},l_{\text proc},\rho_{\text proc})) \mapsto id_{\text proc},l_{\text proc},\rho_{\text proc})~\cdot&\\
    &(f : ((id_{\text coll},l_{\text coll},\rho_{\text coll}), (id_{\text proc},l_{\text proc},\rho_{\text proc})) \mapsto &\\
    &~~id_{\text coll},l_{\text coll},\rho_{\text coll}[{\tt recv\_data} \gets \rho_{\text coll}({\tt recv\_data})~\sem{{\tt op}}~\rho_{\text proc}({\tt send\_data})])~\cdot&\\
    &{\tt Id}^*&
  \end{flalign*}}
\vspace{-1.5em}
  \item \parbox[t]{\textwidth}{
   \vspace{-2.3em} \begin{flalign*}
    G_3 &= ( \top \times \texttt{[MPI_Reduce]}_{\text{lock}} \times \top )^*~\cdot~({\tt root} \times \texttt{[MPI_Reduce]}_{\text{lock}} \times \top)~\cdot&\\
    &( \top \times \texttt{[MPI_Reduce]}_{\text{lock}} \times \top )^*~\cdot~(\top \times \texttt{[Collector]} \times \top)&\\
    F_3 &= (F : ((id,l,\rho), \_) \mapsto id, \texttt{[next_loc]}, \rho)&\\
    &(f : ((id_{\text root},l_{\text root},\rho_{\text root}), (id_{\text coll},l_{\text coll},\rho_{\text coll})) \mapsto&\\
    &\qquad id_{\text root},\texttt{[next_loc]},\rho_{\text proc}[{\tt recv\_data} \gets \rho_{\text coll}(\text{recv\_data})])~\cdot&\\
    &(F : ((id,l,\rho), \_) \mapsto id, \texttt{[next_loc]}, \rho)&\\
  \end{flalign*}}
    \vspace{-3em}
\end{enumerate}

\begin{figure}[ht]
  \scriptsize
\begin{subfigure}[b]{\textwidth}
  \begin{tikzpicture}[->,>=stealth',shorten >=1pt,auto,node distance=2cm,
      semithick]
    
    \tikzset{initial text={}}
    \tikzstyle{every state}=[]
    
    \node[initial left, state, minimum size=15pt] (A) {$q_0$}; 
    
    \node[state, inner sep=0pt, minimum size=15pt] (B) [right=4.2cm of A] {$q_1$};
    
    \node[accepting, state, inner sep=0pt, minimum size=15pt] (C) [right=4.2cm of
      B] {$q_2$};

    \path (A) edge node [xshift=-10pt, yshift=7pt] {$[0,0], \texttt{[MPI_Reduce]},\{ \text{res} \mapsto [\frac{1}{2},\frac{1}{2}] \} $} (B)
    (B) edge  node [xshift=10pt, yshift=7pt] {$[0,0], \texttt{[MPI_Reduce]},\{ \text{res} \mapsto [\frac{1}{4},\frac{1}{4}] \} $} (C);
  \end{tikzpicture}
  \caption{$A$}
\end{subfigure}

\begin{subfigure}[b]{\textwidth}
  \begin{tikzpicture}[->,>=stealth',shorten >=1pt,auto,node distance=2cm,
      semithick]
    
    \tikzset{initial text={}}
    \tikzstyle{every state}=[]

    \node[initial left, state, minimum size=15pt] (A) {$q_0$}; 
    
    \node[state, inner sep=0pt, minimum size=15pt] (B) [right=3cm of A] {$q_1$};
    
    \node[state, inner sep=0pt, minimum size=15pt] (C) [right=3cm of
      B] {$q_2$};

    \node[accepting, state, inner sep=0pt, minimum size=15pt] (D) [right=3cm of
      C] {$q_3$};

    \path 
    (A) edge  node [xshift=-5pt, yshift=8pt] {$[-1,-1], \texttt{[Coll]},\{ \text{total} \mapsto [0,0] \} $} (B)
    (B) edge node [below, xshift=5pt, yshift=-8pt] {$[0,0], \texttt{[MPI_Reduce]}_\text{lock},\{ \text{res} \mapsto [\frac{1}{2},\frac{1}{2}] \} $} (C)
    (C) edge  node [xshift=-5pt, yshift=8pt] {$[1,1], \texttt{[MPI_Reduce]}_\text{lock},\{ \text{res} \mapsto [\frac{1}{4},\frac{1}{4}] \} $} (D);
  \end{tikzpicture}
  \caption{$R(A)$}
\end{subfigure}

\begin{subfigure}[b]{\textwidth}
  \begin{tikzpicture}[->,>=stealth',shorten >=1pt,auto,node distance=2cm,
      semithick]
    
    \tikzset{initial text={}}
    \tikzstyle{every state}=[]

    \node[initial left, state, minimum size=15pt] (A) {$q_0$}; 
    
    \node[state, inner sep=0pt, minimum size=15pt] (B) [right=3cm of A] {$q_1$};
    
    \node[state, inner sep=0pt, minimum size=15pt] (C) [right=3cm of
      B] {$q_2$};

    \node[accepting, state, inner sep=0pt, minimum size=15pt] (D) [right=3cm of
      C] {$q_3$};

    \path 
    (A) edge node [yshift=8pt] {$[0,0], \texttt{[MPI_Reduce]}_\text{lock},\{ \text{res} \mapsto [\frac{1}{2},\frac{1}{2}] \} $} (B)
    (B) edge node [below, yshift=-8pt]  {$[-1,-1], \texttt{[Coll]},\{ \text{total} \mapsto [\frac{1}{2},\frac{1}{2}] \} $} (C)
    (C) edge  node [yshift=8pt] {$[1,1], \texttt{[MPI_Reduce]}_\text{lock},\{ \text{res} \mapsto [\frac{1}{4},\frac{1}{4}] \} $} (D);
  \end{tikzpicture}
  \caption{$R^2(A)$}
\end{subfigure}

\begin{subfigure}[b]{\textwidth}
  \begin{tikzpicture}[->,>=stealth',shorten >=1pt,auto,node distance=2cm,
      semithick]
    
    \tikzset{initial text={}}
    \tikzstyle{every state}=[]

    \node[initial left, state, minimum size=15pt] (A) {$q_0$}; 
    
    \node[state, inner sep=0pt, minimum size=15pt] (B) [right=3cm of A] {$q_1$};
    
    \node[state, inner sep=0pt, minimum size=15pt] (C) [right=3cm of
      B] {$q_2$};

    \node[accepting, state, inner sep=0pt, minimum size=15pt] (D) [right=3cm of
      C] {$q_3$};

    \path 
    (A) edge node [yshift=8pt] {$[0,0], \texttt{[MPI_Reduce]}_\text{lock},\{ \text{res} \mapsto [\frac{1}{2},\frac{1}{2}] \} $} (B)
    (B) edge  node [below, yshift=-8pt] {$[1,1], \texttt{[MPI_Reduce]}_\text{lock},\{ \text{res} \mapsto [\frac{1}{4},\frac{1}{4}] \} $} (C)
    (C) edge node [yshift=6pt]  {$[-1,-1], \texttt{[Coll]},\{ \text{total} \mapsto [\frac{3}{4},\frac{3}{4}] \} $} (D);
  \end{tikzpicture}
  \caption{$R^3(A)$}
\end{subfigure}

\begin{subfigure}[b]{\textwidth}
  \begin{tikzpicture}[->,>=stealth',shorten >=1pt,auto,node distance=2cm,
      semithick]
    
    \tikzset{initial text={}}
    \tikzstyle{every state}=[]
    
    \node[initial left, state, minimum size=15pt] (A) {$q_0$}; 
    
    \node[state, inner sep=0pt, minimum size=15pt] (B) [right=4.2cm of A] {$q_1$};
    
    \node[accepting, state, inner sep=0pt, minimum size=15pt] (C) [right=4.2cm of
      B] {$q_2$};

    \path (A) edge node [xshift=-15pt, yshift=8pt] {$[0,0], \texttt{[MPI_Finalize]},\{ \text{res} \mapsto [\frac{1}{2},\frac{1}{2}, \text{total} \mapsto [\frac{3}{4},\frac{3}{4}] \} $} (B)
      (B) edge  node [xshift=15pt, yshift=8pt] {$[1,1], \texttt{[MPI_Finalize]},\{ \text{res} \mapsto [\frac{1}{4},\frac{1}{4}] \} $} (C);
  \end{tikzpicture}
  \caption{$R^4(A)$}
\end{subfigure}
\caption{\label{fig:reduceapp}Application of $R$ on a configuration}
\end{figure}
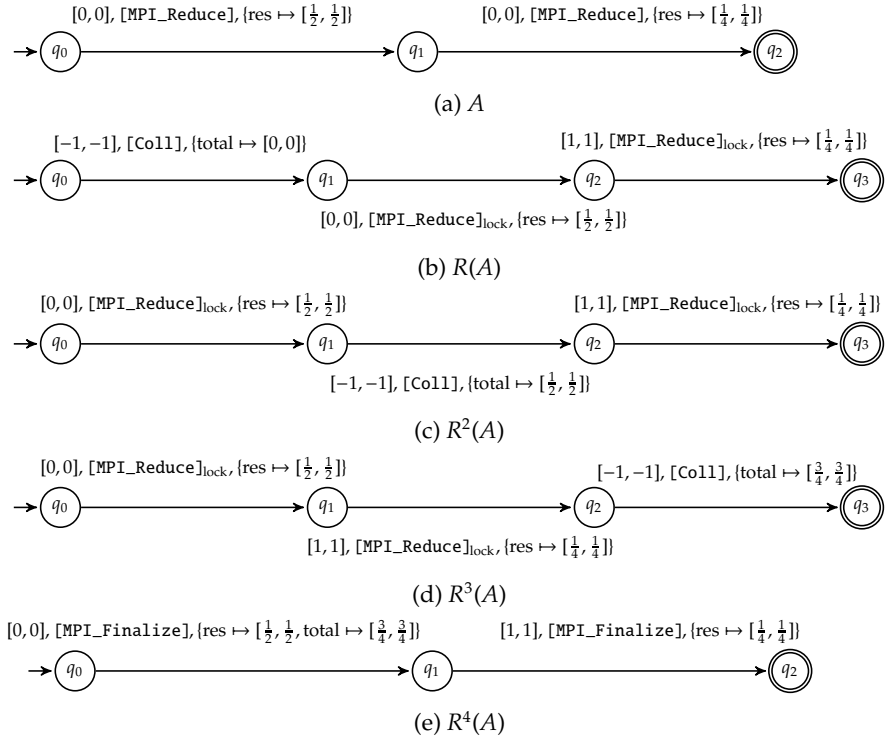

For our example program, our prototype automatically instantiates
these rules in a set $R$. \fref{fig:reduceapp},
illustrates the iterative applications of $R$ on program state $A$
where both processes have reached the \texttt{MPI_Reduce} location by
successive application of the transducer $T$ on the initial
configuration $I$.

We give in \fref{fig:tstar} the full reachability set computed by our
prototype. For readability purposes, we do not show variables that are
not set (i.e. when their value is $\top$). We also factorize the
transitions: for each multiple transitions from a node $p$ to $q$, we
merge them into a single one and concatenate their labelled local
process states (i.e. $(p,\sigma,q);(p,\sigma',q); \Rightarrow (p,
\sigma ; \sigma', q)$).

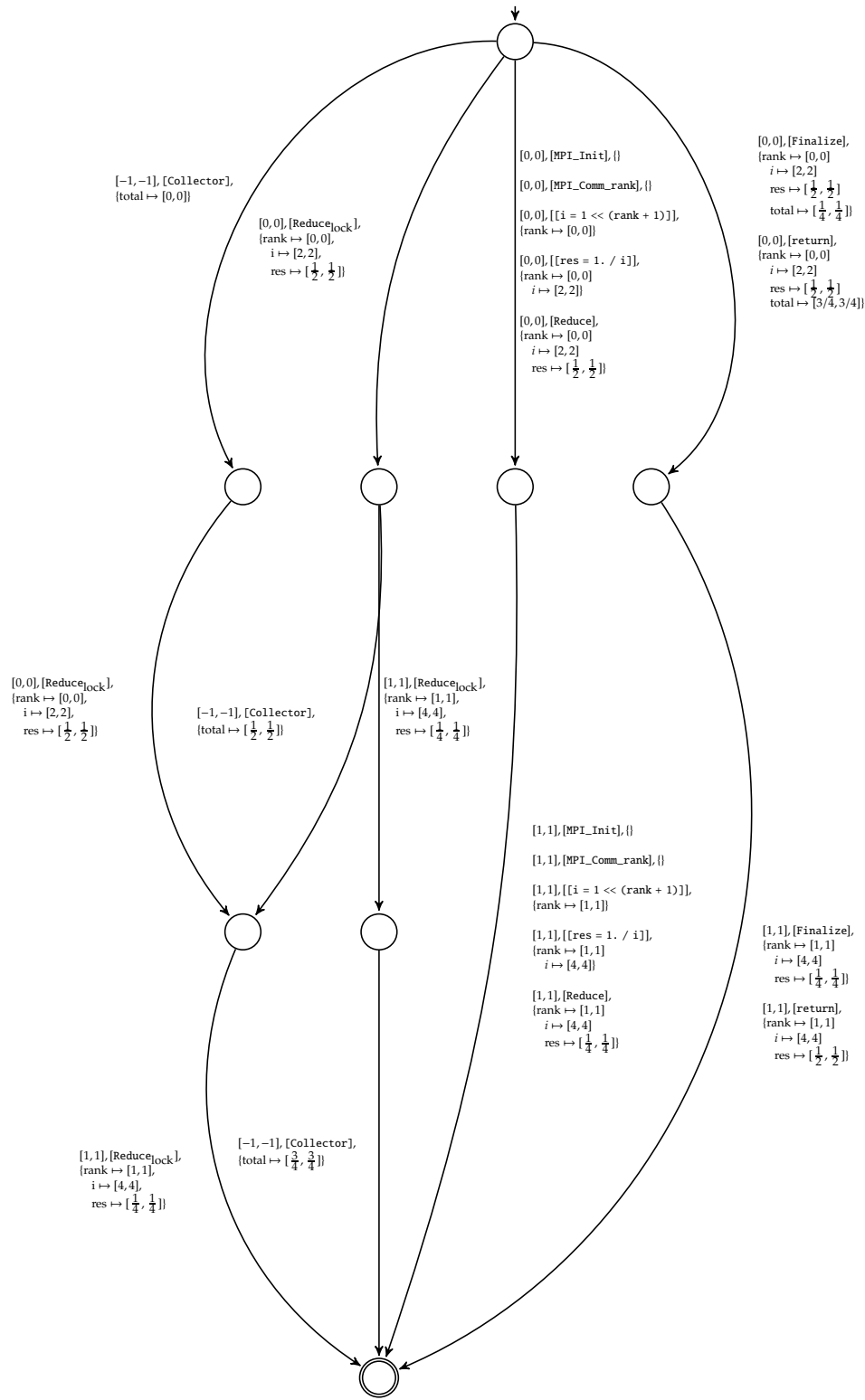
\begin{figure}[ht]
  \centering
  \tiny
  \begin{tikzpicture}[->,>=stealth',shorten >=1pt,auto,node distance=2cm,
      semithick]
    \tikzset{initial text={}}
    \tikzstyle{every state}=[]

    \node[initial above, state, minimum size=15pt] (A) {};
    \node[state, minimum size=15pt] (B3) [below=6cm of A] {};
    \node[state, minimum size=15pt] (B2) [left of=B3] {};
    \node[state, minimum size=15pt] (B1) [left of=B2] {};
    \node[state, minimum size=15pt] (B4) [right of=B3] {};
    \node[state, minimum size=15pt] (C2) [below=6cm of B2] {};
    \node[state, minimum size=15pt] (C1) [left of=C2] {};

    \node[accepting, state, minimum size=15pt] (D) [below=6cm of C2] {};

    \path (A) edge [bend right=60] node [left, text width=2cm]
          {$[-1,-1], \texttt{[Collector]},$\\ 
            $\{ \text{total} \mapsto [0,0] \}$} (B1)

    (A) edge [bend right=20] node [left, text width=2cm]
                {$[0,0], [\texttt{Reduce}_{\text{lock}}],$\\
                  $\{ \text{rank} \mapsto [0,0],$\\
                  $\quad \text{i} \mapsto [2,2],$\\
                  $\quad \text{res} \mapsto [\frac{1}{2},\frac{1}{2}]\} $} (B2)
                
    (A) edge [] node [right, text width=3cm] 
                {$[0,0], [\texttt{MPI\_Init}],\{\}$\\
                  \quad\\
                  $[0,0], [\texttt{MPI\_Comm\_rank}],\{\}$\\
                  \quad\\
                  $[0,0], [\texttt{[i = 1 <{}< (rank + 1)]}],$\\
                  $\{ \text{rank} \mapsto [0,0]\}$\\
                  \quad\\
                  $[0,0], [\texttt{[res = 1. / i]}],$\\
                  $\{ \text{rank} \mapsto [0,0]$\\
                  $\quad i \mapsto [2, 2]\}$\\
                  \quad\\
                  $[0,0], [\texttt{Reduce}],$\\
                  $\{ \text{rank} \mapsto [0,0]$\\
                  $\quad i \mapsto [2, 2]$\\
                  $\quad \text{res} \mapsto [\frac{1}{2}, \frac{1}{2}] \}$
                } (B3)

    (A) edge [bend left=70] node [right, xshift=0.5cm, text width=3cm] 
                {$[0,0], [\texttt{Finalize}],$\\
                  $\{ \text{rank} \mapsto [0,0]$\\
                  $\quad i \mapsto [2,2]$\\
                  $\quad \text{res} \mapsto [\frac{1}{2}, \frac{1}{2}]$\\
                  $\quad \text{total} \mapsto [\frac{1}{4}, \frac{1}{4}] \}$\\
                  \quad\\
                  $[0,0], [\texttt{return}],$\\
                  $\{ \text{rank} \mapsto [0,0]$\\
                  $\quad i \mapsto [2, 2]$\\
                  $\quad \text{res} \mapsto [\frac{1}{2}, \frac{1}{2}]$\\
                  $\quad \text{total} \mapsto [3/4, 3/4] \}$
                } (B4)      

    (B4) edge [bend left=50] node [right, xshift=0.5cm, text width=3cm] 
                {$[1,1], [\texttt{Finalize}],$\\
                  $\{ \text{rank} \mapsto [1,1]$\\
                  $\quad i \mapsto [4, 4]$\\
                  $\quad \text{res} \mapsto [\frac{1}{4}, \frac{1}{4}]\}$\\
                  \quad\\
                  $[1,1], [\texttt{return}],$\\
                  $\{ \text{rank} \mapsto [1,1]$\\
                  $\quad i \mapsto [4, 4]$\\
                  $\quad \text{res} \mapsto [\frac{1}{2}, \frac{1}{2}]\}$
                } (D)      

    (B3) edge [bend left=10] node [right, xshift=0.5cm, text width=3cm] 
                {$[1,1], [\texttt{MPI\_Init}],\{\}$\\
                  \quad\\
                  $[1,1], [\texttt{MPI\_Comm\_rank}],\{\}$\\
                  \quad\\
                  $[1,1], [\texttt{[i = 1 <{}< (rank + 1)]}],$\\
                  $\{ \text{rank} \mapsto [1,1]\}$\\
                  \quad\\
                  $[1,1], [\texttt{[res = 1. / i]}],$\\
                  $\{ \text{rank} \mapsto [1,1]$\\
                  $\quad i \mapsto [4, 4]\}$\\
                  \quad\\
                  $[1,1], [\texttt{Reduce}],$\\
                  $\{ \text{rank} \mapsto [1,1]$\\
                  $\quad i \mapsto [4, 4]$\\
                  $\quad \text{res} \mapsto [\frac{1}{4}, \frac{1}{4}]\}$
                } (D)      

    (B1) edge [bend right=40] node [left, text width=2cm]
                {$[0,0], [\texttt{Reduce}_{\text{lock}}],$\\
                  $\{ \text{rank} \mapsto [0,0],$\\
                  $\quad \text{i} \mapsto [2,2],$\\
                  $\quad \text{res} \mapsto [\frac{1}{2},\frac{1}{2}] \} $
                } (C1)

    (B2) edge [bend left=20] node [left,text width=2.3cm]
                {$[-1,-1], \texttt{[Collector]},$\\ 
                  $\{ \text{total} \mapsto [\frac{1}{2},\frac{1}{2}]\}$} (C1)

    (B2) edge node [right, text width=2cm]
                {$[1,1], [\texttt{Reduce}_{\text{lock}}],$\\
                  $\{ \text{rank} \mapsto [1,1],$\\
                  $\quad \text{i} \mapsto [4,4],$\\
                  $\quad \text{res} \mapsto [\frac{1}{4},\frac{1}{4}] \} $
                } (C2)

    (C1) edge [bend right=40] node [left, text width=2cm]
                {$[1,1], [\texttt{Reduce}_{\text{lock}}],$\\
                  $\{ \text{rank} \mapsto [1,1],$\\
                  $\quad \text{i} \mapsto [4,4],$\\
                  $\quad \text{res} \mapsto [\frac{1}{4},\frac{1}{4}] \}$
                } (D)

    (C2) edge node [left, text width=2cm]
                {$[-1,-1], \texttt{[Collector]},$\\ 
                  $\{ \text{total} \mapsto [\frac{3}{4},\frac{3}{4}]\}$
                } (D);
  \end{tikzpicture}  
  \caption{\label{fig:tstar}Sum program reachability set with 2 processes}
\end{figure}

\section {Experiments}
\label{sec:exp}

We present in this section some of the analysis results of our
prototype. We found several tools that provide a formal verification
of MPI programs. One of the most advanced we found is called ``In-situ
Partial Order''~\cite{vo2009formal} (ISP). It is based on model
checking and performs a dynamic analysis in order to detect the
presence of \emph{deadlocks}. To the best of our knowledge, our tool
is the only one that computes the reachability set. We present some
examples where we verify numerical properties and although our
prototype focuses on safety properties, it can also detects that
program states (i.e. words) in our set are not matched by any rules
and therefore detect deadlocks. In these cases, we can raise an alarm
(which can be false ones due to our abstractions).

\begin{figure}[H]
  
  \centering
  \begin{tabular}{c|c|c|c|c|c|c}
    Program & LoC & nb proc. & state space size & nb nodes & nb transitions & exec. time\\
    \hline
    deadlock random & 23 & 2 & 225 & 4 & 20 & 0.5s \\

    dining philosophers & 42 & 4 & 83521 & 17 & 112 & 8s \\

    dining philosophers & 42 & 6 & 24137569 & 54 & 447 & 1232s \\
   
    \hline\hline

    sum program & 11 & 50 & $\infty$ & 200 & 601 & 86s\\

    pi approximation & 26 & 50 & $\infty$ & 200 & 751 & 104s
  \end{tabular}
  
\end{figure}


We tested our prototype on several examples. We display here the
results of significant ones. Our parameters are: the number of
processes we start with, the concrete state space size, the number of
nodes in the lattice automaton that  represents the final reachability set, its number of transitions and the
execution time. The concrete state space size is the enumeration of
all possible program states; it is infinite when there are integer variables. We prove on these examples two kinds of
properties: deadlock detection and numerical safety properties.

First is a potentially deadlocking program ``random deadlock'' where
two process tries to communicates randomly: both test a random
condition that leads respectively to a send or a receive call towards
the other process. As ISP is dynamic and depends on the MPI execution,
it will not always detect this simple deadlock. However, as we compute
the reachability set, we easily observe this deadlock and can raise an
alarm.

We implemented a MPI version of the dining philosopher problem where
philosophers and forks are processes. The forks processes will give
permission to ``pick them up'' and ``put them down'' modeled by
point-to-point communications. 
Naturally, the program has deadlocks and again the
reachability set exhibits them. The growth in computation time is
explained by the amount of possible interleavings that our algorithm
is currently not capable of filtering and by the precision we wish to
attain (thus, no strong abstractions) in order to precisely determine
the deadlocks (and not a false alarm).

The next two following examples both implement a floating point value
approximation. The first one is our
example program used in~\fref{fig:simpleprog}. The same property is
used: $\text{total} \in [0,1]$. The second one is a computation on pi
based on the approximation of $\int_{0}^{1} \frac{4}{1+x^2}$ with sums
of $n$ intervals dispatched on $n$ processes. Again the property is a
framing of the result ($\in [3,4]$) . These two examples display the
capacity of our prototype to handle real-life computations. However,
we would like to generalize these two examples to any number of
process. We can model an initial configuration with an unbounded
number of process and run our analysis on it. Unfortunately, we cannot
infer a relation between the process rank or the number of processes
with our current numerical domains. Therefore, our \emph{sum
  program}'s analysis, on an unbounded number of process, can detect
that each process computes a local $\text{result} \in [0,\frac{1}{2}]$
but the total sum will be abstracted to $[0,+\infty]$.

\section{Conclusion}
\label{sec:conclusion}

We presented a new way to do static analysis on a model of concurrent
programs that allow unicast and multicast communication as well as
dynamic process creation. We described the general framework of the
method with well-founded abstraction of the semantics and program
states. We applied our technique in order to compute reachability sets
of MPI concurrent programs with numerical abstract domains. We showed
that building such an analysis on a realistic language, such as MPI/C
programs, is feasible and yields encouraging results. Moreover,
abstract interpretation allows us to verify numerical properties which
was not done before on such programs, and the lattice automata allow
the analysis to represent (and automatically discover) regular
invariants on the whole program states.

Future work includes theoretical and practical improvements of our
analyzer,  especially the application algorithm which is currently not
optimized. One way to do that is to run a quick pre-analysis using a
simple, non-numerical abstract domain to obtain information
(e.g. rewriting rules that are never activated), so that we may
simplify the rules before using more costly numerical abstract
domains. We also wish to design a specification language allowing us
to write regular properties more easily.
We will also improve our analyser by taking into account more MPI
primitives as well as supporting general C constructs (pointers,
functions, etc.) thanks to better interactions with the other Frama-C
plug-ins. Finally, we will deal with asynchronous communications (FIFO
queues) and shared variables using non-standard semantics and/or a
reduced product with abstract domains that can efficiently abstract
these kind of data.

\bibliographystyle{plain} \bibliography{article}

\end{document}